%% file: sourcefile_manuscript_adjusted.tex
\documentclass[10pt, twocolumn]{revtex4-2}

\usepackage{array}

\usepackage[utf8]{inputenc}
\usepackage[T1]{fontenc}

\usepackage{amsmath, amssymb, amsthm, mathtools}
\usepackage{bm}
\usepackage{braket}
\usepackage{bbold}
\usepackage{siunitx}
\usepackage{cancel}
\usepackage{tensor}
\usepackage{physics}

\usepackage{graphicx}
\usepackage{float}
\usepackage{subcaption}
\usepackage{tikz, pgfplots}
\pgfplotsset{compat=1.18}
\usetikzlibrary{
  positioning, fit, arrows.meta, calc,
  shapes.geometric, arrows,
  decorations.markings
}

\usepackage{xcolor}

\usepackage{enumitem}
\setlist[enumerate]{label=\arabic*.}
\newlist{subenumerate}{enumerate}{2}
\setlist[subenumerate,1]{label=\arabic{subenumeratei}.\arabic*}
\setlist[subenumerate,2]{label=\arabic{subenumeratei}.\arabic{subenumerateii}}

\usepackage{hyperref}
\hypersetup{colorlinks=true, linkcolor=blue, citecolor=blue, urlcolor=blue}
\usepackage{cleveref}

\usepackage{datetime}
\newdateformat{monthyeardate}{\monthname[\THEMONTH] \THEYEAR}

\usepackage{caption}
\tikzstyle{startstop} = [circle, minimum width=3cm, minimum height=1cm,
  text centered, draw=black, fill=blue!30!white]
\tikzstyle{q} = [circle, minimum width=3cm, minimum height=1cm,
  text centered, draw=black, fill=blue!40!white]
\tikzstyle{arrow} = [thick,->,>=stealth]

\definecolor{errPresent}{RGB}{45,53,97}
\definecolor{errCaveat}{RGB}{122,141,192}
\definecolor{errAbsentBg}{RGB}{244,241,234}
\definecolor{errSectionBg}{RGB}{238,234,224}
\definecolor{errHeaderBg}{RGB}{26,26,46}

\newcommand{\symP}{%
  \tikz[baseline=-0.6ex]
    \filldraw[fill=errPresent, draw=none] (0,0) circle (4pt);%
}
\newcommand{\symA}{%
  \tikz[baseline=-0.6ex]
    \filldraw[fill=errAbsentBg, draw=black!30, line width=0.5pt]
             (0,0) circle (4pt);%
}
\newcommand{\symC}{%
  \tikz[baseline=-0.6ex]
    \filldraw[fill=errCaveat, draw=none] (0,0) circle (4pt);%
}
\newcommand{\pillP}[1]{%
  \tikz[baseline=-0.6ex]
    \node[fill=errPresent, text=white, rounded corners=3pt,
          inner xsep=4pt, inner ysep=2pt, font=\scriptsize\rmfamily]
         {$\times$#1};%
}
\newcommand{\pillC}[1]{%
  \tikz[baseline=-0.6ex]
    \node[fill=errCaveat, text=white, rounded corners=3pt,
          inner xsep=4pt, inner ysep=2pt, font=\scriptsize\rmfamily]
         {$\times$#1};%
}

\newcolumntype{A}{>{\raggedright\arraybackslash}p{3cm}}
\newcolumntype{E}{>{\centering\arraybackslash}m{1.87cm}}
\newcolumntype{F}{>{\centering\arraybackslash}m{2cm}}

\definecolor{emgreen}{RGB}{72, 160, 72}
\definecolor{meblue}{RGB}{48, 122, 185}
\definecolor{rorange}{RGB}{215, 110, 45}
\definecolor{lneutral}{RGB}{90,105,135}   %

\tikzset{
  foldtri neutral/.style={
    regular polygon, regular polygon sides=3,
    fill=lneutral!50, draw=lneutral!80!black,
    inner sep=2.2pt, shape border rotate=-90,
    font=\tiny
  },
}

\def\hind{1.7}

\def\Rx{7.8}

\tikzset{
  mainbox/.style={rounded corners=7pt, draw=#1!85!black, line width=1.2pt,
    fill=#1, text=white, font=\bfseries\large,
    minimum width=3.6cm, minimum height=0.85cm, align=center, inner sep=7pt},
  subbox/.style={rounded corners=4pt, draw=#1!70!black, line width=0.8pt,
    fill=#1!35, text=#1!25!black, font=\small,
    minimum width=3.0cm, minimum height=0.62cm, align=center, inner sep=5pt},
  leafbox/.style={rounded corners=4pt, draw=rorange!70!black, line width=0.8pt,
    fill=rorange!30, text=rorange!25!black, font=\small,
    minimum width=3.0cm, minimum height=0.62cm, align=center, inner sep=5pt},
  arr/.style={-{Stealth[length=5pt, width=4pt]}, line width=1pt},
  garr/.style={arr, color=emgreen!85!black},
  barr/.style={arr, color=meblue!80!black},
  oarr/.style={arr, color=rorange!85!black},
  dashbox/.style={draw=gray!60, dashed, line width=0.9pt,
    rounded corners=5pt, inner sep=9pt},
}

\theoremstyle{definition}

\theoremstyle{plain}
\newtheorem{theorem}{Theorem}[section]
\newtheorem{lemma}[theorem]{Lemma}
\newtheorem{corollary}[theorem]{Corollary}
\newtheorem{proposition}[theorem]{Proposition}
\theoremstyle{remark}

\newcommand{\EE}{\mathsf{E}}
\newcommand{\PP}{\mathsf{P}}
\newcommand{\indicator}{\mathsf{1}}

\newcommand{\NN}{\mathbb{N}}

\newcommand{\kvec}{\boldsymbol{k}}
\newcommand{\Geo}{\mathrm{Geo}}
\newcommand{\TruncGeo}{\mathrm{TruncGeo}}

\newcounter{protocol}

\begin{document}

\onecolumngrid
\begingroup
\centering
\Large \textbf{Verifiable blind quantum computing: Comparative analysis and design considerations for client architectures}\\[1.5em]
\large J van Dam$^{1,2,3,*}$, J Grimbergen$^{1,2,3}$ and SDC Wehner$^{1,2,3}$\\
\endgroup
\vspace{10pt}
\small \noindent \par
$^1$ QuTech, Delft University of Technology, Lorentzweg 1, 2628 CJ, Delft, The Netherlands\\
$^2$ Kavli Institute of Nanoscience, Delft University of Technology, Lorentzweg 1, 2628 CJ, Delft, The Netherlands\\
$^3$ Quantum Computer Science, EEMCS, Delft University of Technology, Lorentzweg 1, 2628 CJ, Delft, The Netherlands\\

$^*$ Corresponding author: j.vandam-3@tudelft.nl\\

\begin{abstract}
Blind quantum computing (BQC) allows a client to delegate quantum computations to a remote server without revealing the input, computation, or output. In addition to being blind, the client can sometimes also verify that the server has performed their instructions correctly, a property known as verifiability. A key part of realizing such verifiable BQC (VBQC) is choosing the design of the client device: many architectures have been proposed, each with different hardware requirements, security properties, and performance characteristics, making it difficult to identify which is most suitable for a given implementation. In this work, we present a comparative analysis of client architectures for VBQC with a matter-qubit server. We restrict our analysis to single-server, single-client protocols with information-theoretic security based on measurement-based quantum computation. We identify three main categories of client: emission-based, measurement-based, and rotation-based, each with multiple variants depending on how the client interacts with the server. We evaluate each across different dimensions: we compare guarantees of existing corresponding security proofs, we derive equations for the rate at which each client can execute a protocol, we provide an overview of each architecture's error behaviour, and discuss hardware cost and design considerations. Client architectures implementing measurement-based remote state preparation and reflection-based teleportation emerge as strong default candidates, but as the right choice remains context-dependent, we provide a framework for navigating considerations to guide the selection of the most suitable architecture for a given setting.
\end{abstract}
\maketitle

\section{Introduction}

In blind quantum computing (BQC), a client can execute quantum algorithms at a remote quantum server without the input, the computation, or its outcome being revealed, apart from an upper bound on the size of the computation~\cite{broadbent2009universal}. In this setup, a quantum server is assumed to come with a large financial cost, whereas the client holds a relatively simple and cheap device. This way, the quantum computer can be used by a much wider audience. In addition, some protocols allow the client to verify the server: allowing a client to detect deviations from a known protocol, such that it can detect errors or malicious behaviour from the server. This is referred to as verifiable blind quantum computing (VBQC) \cite{fitzsimons2017unconditionally}. 

The client device in VBQC can be realized in different ways. Each client architecture comes with different considerations and trade-offs regarding the complexity of the implementation, security considerations and performance. These variations offer a practical challenge: which architecture should actually be implemented? The answer to this question will involve both performance analysis and real world considerations such as hardware costs, security requirements and fabrication challenges.

In this work, we present a comparative analysis of client architectures for VBQC with a matter-qubit server. We consider only single-server single-client protocols, excluding schemes that rely on multiple non-communicating servers~\cite{morimae2013secure, sheng2015deterministic, li2014triple, quan2023verifiable} or multiple clients such as in the Qline architecture~\cite{qline}. While non-communicating server protocols enable fully classical clients, not allowing untrusted servers to communicate is difficult to enforce in practice. Second, we focus on servers built from matter qubits (such as trapped ions \cite{bruzewicz2019trapped}, NV centers \cite{doherty2013nitrogen}, or neutral atoms \cite{henriet2020quantum}) that interface with photonic qubits for communication. Third, we require information-theoretic security, excluding protocols whose blindness relies on computational assumptions such as the hardness of learning with errors or the existence of trapdoor claw-free functions~\cite{mahadev2018classical, broadbent2015quantum}. While computationally secure protocols such as QEnclave~\cite{ma2022qenclave} offer the appealing prospect of a fully classical client, their security guarantees may be compromised by future algorithmic advances or the development of sufficiently powerful quantum computers. Finally, we restrict our analysis to protocols based on measurement-based quantum computation (MBQC)~\cite{raussendorf2001one}, where computations are carried out by single-qubit measurements on a graph state resource~\cite{hein2004multiparty}. This restriction follows naturally from the previous constraints in combination with the desire to have a simple client device: MBQC-based protocols achieve information-theoretic security with minimal client requirements, whereas circuit-based alternatives with information-theoretic security typically require the client to perform multi-qubit operations~\cite{childs2005secure}.

Our scope leaves us with three primary client categories: (1) emission-based clients, (2) measurement-based clients, and (3) rotation-based clients. These architectures then have sub-categories on how they interact with the server. We first lay out the client architecture taxonomy, discussing the implementation options across these categories. Then, to compare these architectures, we combine quantitative performance analysis with practical experimental considerations: for each architecture, we evaluate comparative hardware cost, technical complexities and security implications, drawing on consultations with experimental researchers and security experts, among other sources. We also provide an overview of the errors occurring in each system, and derive equations for their rate of execution. While not exhaustive, this hybrid approach to evaluation allows us to capture important considerations and allows us to showcase strengths and weaknesses for each architecture, to aid coming to practical conclusions for implementation. In this analysis two architectures stand out as promising default candidates: a client architecture implementing measurement-based remote state preparation and an emission-based client using cavity reflection to remotely prepare states. These candidates perform well in terms of rate, have relatively few error modes, and are compatible with strong security protocols. Even so, the best choice remains setting-dependent. It depends on what a given group values most --- whether that is security, execution rate or fidelity, or hardware cost --- and on the practical realities of the lab: existing expertise with a particular platform, or simply whether the necessary components, such as an optical cavity, are available. Our analysis is therefore intended as a decision-making tool, structured to help navigate these trade-offs in a given setting.

\section{Client architecture taxonomy}\label{sec:taxonomy}

Within the scope presented in the introduction, we will introduce the client architectures that are included in this comparative analysis. There are different ways of categorizing the clients, either by the \textit{role} the client fulfils in the protocol or by the \textit{capabilities} the client has. Below, we introduce both categorizations. The role categories are `prepare-and-send' and `receive-and-measure'. The client capabilities can be split into three main categories: client devices that \textit{emit}, \textit{measure} or \textit{rotate} photonic qubits. Within each of those categories, there are multiple variations, which we will cover in the subsections below. An overview of the taxonomy is depicted in Figure~\ref{fig:taxonomy}.

\begin{figure*}
    \centering
    \input{client_taxonomy.tex}
    \caption{Overview of client architecture taxonomy. Three categories based on client capabilities: Emission-based (green), measurement-based (blue) and rotation-based (orange). The clients within these categories can belong to sub categories, such as prepare-and-send and receive-and-measure, indicating the role they play in the VBQC protocol. Additionally, the rotation-based clients are separated by how often a signal needs to travel between the client and the server (single-pass vs multi-pass), and by the hardware that implements them (`fixed-gate' vs non-fixed). Triangles indicate alternative photon-matter interfaces available for the corresponding client, as shown in the legend in the bottom.}
    \label{fig:taxonomy}
\end{figure*}
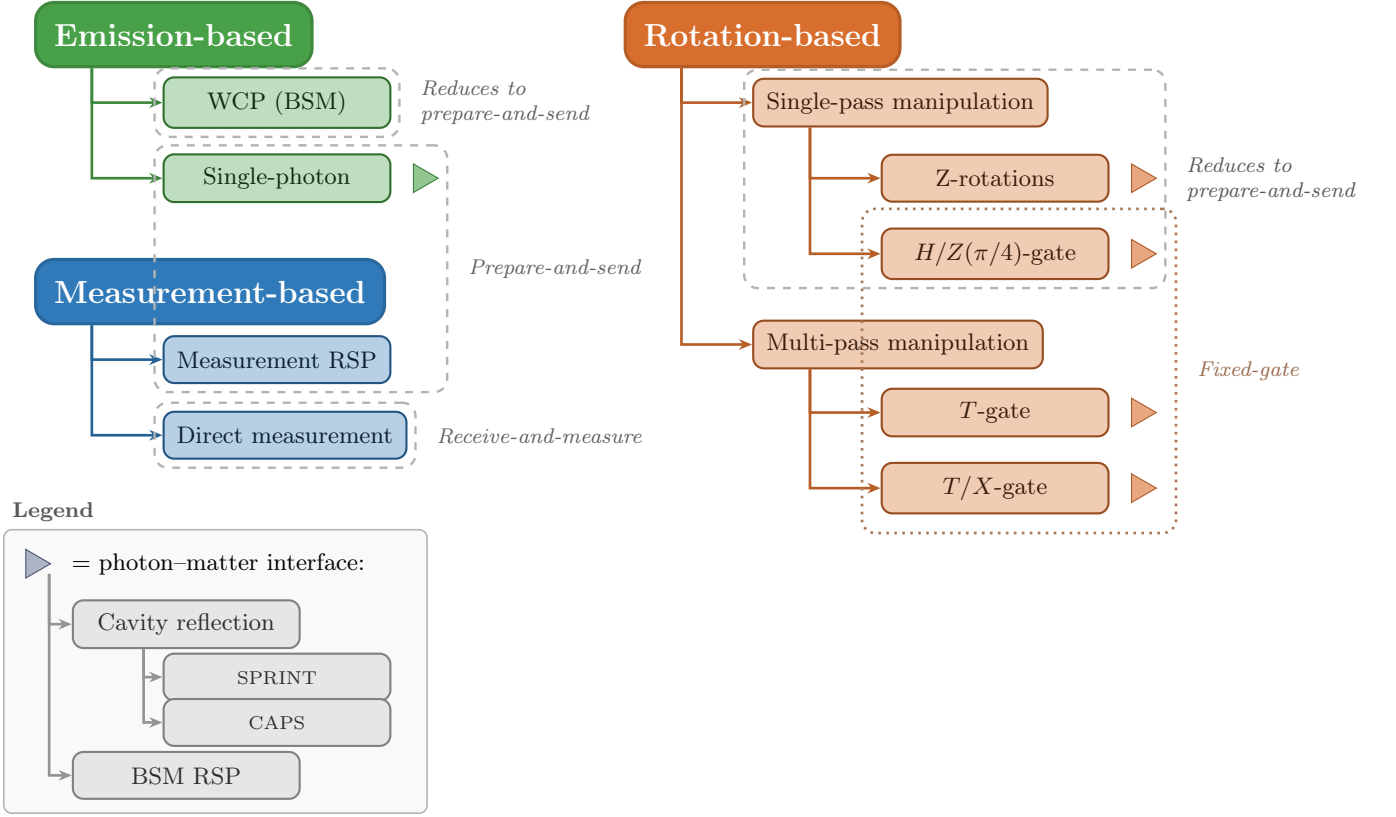

\subsection{Client roles: prepare-and-send or receive-and-measure}
To understand the distinction between client roles, we first briefly recall the computational model that underlies the protocols in our scope. In MBQC~\cite{raussendorf2001one}, a computation is performed by preparing a large entangled resource state, consisting of qubits on the $X/Y$ plane of the Bloch sphere (i.e., of the form $|\pm_\theta\rangle=(|0\rangle+e^{i\theta}|1\rangle)/\sqrt{2}$). Then the computation is executed entirely through adaptive single-qubit measurements also in the $X/Y$ plane. The computation is fully determined by the choice of measurement bases: each qubit in the entangled resource state is measured at an angle $\phi$ that depends on the desired gate and, in general, on the outcomes of previous measurements. Blindness in BQC therefore reduces to hiding these measurement angles $\phi$ from the server.

There are two fundamentally different strategies for achieving this, corresponding to the two client roles: prepare-and-send and receive-and-measure~\cite{gheorghiu2019verification}.

In \emph{prepare-and-send} protocols \cite{broadbent2009universal, fitzsimons2017unconditionally}, the client ensures that the server's qubits are initialized in states that are unknown to the server, i.e., the client determines $\theta$ without the server knowing. This is achieved by the client having some sort of control over the server qubits, e.g. by remote state preparation. Here, $\theta$ is randomly chosen from $\{k\pi/4\}^{7}_{k=0}$. The server prepares these qubits in a graph state by performing $CZ$ operations and, when instructed by the client, measures each qubit at a masked angle $\delta$ communicated by the client. This $\delta$ will depend on the secretly chosen $\theta$, the desired, and also secret, computation angle $\phi$, a random bit $r$ also chosen by the client, and some correction based on previous measurement outcomes. Since the server knows only $\delta$ but not $\theta$ or $r$, it cannot determine $\phi$ and thus learns nothing about the computation. The security of this approach therefore relies on the client's ability to control the quantum state at the server without revealing it.

In \emph{receive-and-measure} protocols~\cite{morimae2014verification, hayashi2015verifiable}, the roles are reversed: the server prepares the resource state and coherently sends the qubits to the client, who performs the measurements, and therefore the computation, directly. Here, blindness follows from a different principle. Because the client's measurement choices are made locally and after the server has sent the qubits, the no-signalling principle \cite{eberhard1989quantum} guarantees that the server's quantum state cannot depend on the client's future measurement choices. This makes the blindness guarantee device-independent: it holds without any assumptions on the server's honest behaviour or the quality of the states it sends.

While this categorization says something about the role of the client, it does not necessarily reflect the hardware capabilities of the client, as we will cover below. Specifically, the only client architecture that falls in the receive-and-measure category is the direct measurement client. Client architectures that control the state of the server remotely, including measurement-based remote state preparation (RSP) and teleportation-based approaches, fall under prepare-and-send. This characterization is also depicted in Figure~\ref{fig:taxonomy}. 

Below, we cover the client capability categories, which will be the main structure in the comparative analysis of Section \ref{sec:compare}.

\subsection{Emission-based clients}\label{subsec:emission}

Emission-based clients have the ability to create photonic qubits and send them to the quantum server, e.g.\ over optical fibre. The device consists of a photon source and optical elements to manipulate the qubit state, such as waveplates, liquid crystal retarders or electro-optic modulators (EOMs). The photon source can be either a single-photon source, or weak coherent pulse (WCP) source. The photonic qubit coming from this source is then sent to the server, and is either teleported into the servers memory through a Bell state measurement (BSM), or transferred through an approach that uses reflection off the server's cavity, if it has one. 

\textbf{WCP. } This client type uses a WCP source. It has a low mean photon-number, meaning that it functions as an approximate single photon source. Aside from some practical and security considerations that will be discussed in Sections~\ref{subsec:security} and~\ref{subsec:cost}, it functions otherwise the same as a single-photon based client. The interaction with the server has been studied specifically for the BSM-teleportation case \cite{vandam2025single}. While it is possible that a reflection-based scheme also works in combination with a WCP, this has not been studied specifically and is therefore left out of the comparison for now. 

\textbf{Bell state measurement.} The server emits a photon entangled with its quantum memory. This photon travels to a midpoint station where it interacts with a photonic qubit (or WCP) sent from the client in a BSM, performing teleportation and thereby transferring the client's qubit state onto the server memory qubit. Depending on whether the encoding of the photon state is dual-rail (e.g.\ polarization or time-bin)~\cite{barrett2005efficient} or single-rail (Fock basis)~\cite{campbell2008measurement}, teleportation can be performed using either the double-click or single-click method (i.e., it is heralded upon detection of two or one photons at the BSM). There are currently no security proofs that consider the single-rail encoding, we therefore only focus on the dual-rail/double-click approach. 

\textbf{Single photon cavity reflection.} Instead of transferring the client's prepared state through a BSM, it can be transferred by interacting with the server's cavity, if it has one. This can be done in different ways: one way implements a SWAP gate through single-photon Raman interaction, and a second way implements a teleportation through a $CZ$-gate and $X$-basis measurement. 

The first approach to atom-photon state transfer is offered by single-photon Raman interaction (SPRINT), which exploits the chiral coupling between propagation direction and internal atomic transitions in a whispering-gallery-mode (WGM) microresonator \cite{pinotsi2008single, rosenblum2017analysis}. In this scheme, a three-level $\Lambda$ system is coupled to the two counterpropagating modes of the resonator via their opposite circular polarizations, such that an incoming photon in a superposition of the two modes undergoes complete destructive interference in the forward direction and is reflected, but only via the Raman channel that simultaneously flips the atomic ground state. This interaction is heralded upon detection of the reflected photon. The result is a passive, deterministic SWAP gate between the photonic and atomic qubits, transferring the photon's state directly into the atom. 

The second cavity-reflection approach is the cavity-assisted photon scattering (CAPS) gate, which is based on the conditional phase shift acquired by a polarization-encoded photon reflecting from a one-sided Fabry-Perot cavity containing a three-level atom \cite{duan2004scalable, kikura2025passive}. Here the two polarization components of the photon are routed by a polarizing beam splitter: one reflects off the cavity, acquiring a $\pi$ phase shift conditioned on the atomic state, while the other reflects off a plain mirror. The interaction thus realizes a $CZ$ gate between the photonic and atomic qubits. State transfer is completed in a second step by measuring the photon in the $X$-basis and applying the appropriate Pauli correction, making the protocol a heralded teleportation-based loading scheme rather than a direct swap.

\subsection{Measurement-based clients}\label{subsec:measurement}

Measurement-based clients consist of single photon detectors combined with optical elements. The optical elements are used to manipulate the state of photonic qubits which allows the client to measure incoming photonic qubits in specific bases. Measurements on the $X/Y$ plane of the Bloch sphere (i.e. in a basis defined by $|\pm_\theta\rangle$) often suffice, but sometimes $Z$-basis measurements are also needed. We distinguish two ways in which this client can interact with the server: through measurement-based remote state preparation (RSP) or through direct measurement.

\textbf{Measurement-based RSP.} This type of client is of the prepare-and-send type. The server emits a photon entangled with a memory qubit, this photonic qubit travels to the client. The client measures this photon in a chosen basis, thereby projecting the memory qubit onto that same basis. For example, measuring one half of a Bell state $\ket{\Phi^+} = (\ket{00}+\ket{11})/\sqrt{2}$ in the $\ket{\pm_\theta}$ basis projects the other qubit onto $\ket{\pm_\theta}$. This architecture is also discussed and studied in Ref. \cite{vandam2024hardware}.

\textbf{Direct measurement.} This type of client is the only type which has the receive-and-measure role: The entangled resource state is prepared by the server and sent to the client coherently one by one. The client measures the incoming qubits in the $|\pm_\phi\rangle$ basis, and sometimes in the standard basis depending on the verification strategy. 

\subsection{Rotation-based clients}\label{subsec:rotation}

In this configuration, the server is in charge of creating the single photons. The client is in charge of manipulating this qubit (applying a rotation) and sending it back to the server. We distinguish two sub-categories: single-pass manipulation clients, where the client applies all gates in a single round trip, and multi-pass manipulation clients, where the server must apply gates between multiple client-side operations, requiring multiple round trips. For all rotation-based architectures, the photon that is manipulated needs to be transferred back onto the server's matter qubit, sometimes multiple times. This can be done either through a BSM or using a cavity-reflection approach.

\textbf{Single-pass manipulation clients.} This client type functions very similarly to the single-photon emission client, except that now the photon source itself is not located at the client, but at the server. The client still has optical elements to manipulate the photonic qubit state: the server sends a single photonic qubit to the client, the client rotates the photon qubit state and sends it back to the server. This requires a two-way quantum channel between server and client, where for each qubit a photon needs to do a full round-trip. There are two approaches for this. In the first approach, the server emits a photonic qubit in the state $|+\rangle$ that travels to the client. The client performs a phase rotation and random bit flip using simple linear optics and sends the photonic qubit back to the server. This $Z$-rotation client is introduced in Ref. \cite{kashefi2024verification}.

Alternatively, Li et al.~\cite{li2021blind} proposed a variant in which the client has the fixed gates $H$ and $Z(\pi/4)$ instead of a variable $Z$-rotator and bit flips. In this approach, the server sends $|0\rangle$ to the client instead of $|+\rangle$, and the client applies $H$ followed by $Z(\pi/4)^k$ with a random $k\in\{0, 1, ..., 7\}$, by looping through the $Z(\pi/4)$ gate $k$ times. This produces the required state $|+_\theta\rangle$ with $\theta=k\pi/4$ which is sent back to the server. 

\textbf{Multi-pass manipulation clients.} Wu et al.~\cite{wu2023blind} introduce two approaches that require multiple round-trips. The first option further simplifies the client's capabilities by restricting operations to fixed-angle rotations~\cite{wu2023blind}. In the first variant of this protocol, the server prepares qubits in the $\ket{+}$ state and sends them to the client, who applies a fixed $T = e^{-i\pi/8}Z(\pi/4)$ gate a variable number of times before returning the qubit. The server then applies a Hadamard gate and sends the qubit back for another round of $T$ gate applications, followed by a final server-side Hadamard. The resulting state $HT^jHT^i\ket{+}$ spans the full set $\{\ket{0}, \ket{1}, \ket{k\pi/4}\}_{k=0}^{7}$ required for their protocol, depending on the number of $T$ applications $i$ and $j$ chosen by the client. Since the rotation angle is fixed, the client's quantum hardware reduces to a passive optical element (such as a $\lambda/8$ waveplate) combined with an optical switch that controls the number of passes through this element. 

A second variant in~\cite{wu2023blind} allows the client to additionally perform $X$ gates. In this case, the server prepares qubits in the $\ket{0}$ state, and the protocol proceeds through three communication rounds with the server applying $T$ and $S=T^2$ gates between rounds rather than Hadamards. The qubits are divided into subsets that undergo different sequences of client operations: some are transformed into odd-multiple angles $\ket{k\pi/4}$ for $k \in \{1,3,5,7\}$, others into even-multiple angles $\ket{\pi/2}$ or $\ket{3\pi/2}$, and the remainder into $\ket{+}$, $\ket{-}$, $\ket{0}$, or $\ket{1}$ states. This variant requires an additional round-trip between the client and the server.

\section{Comparative analysis and design considerations}\label{sec:compare}
Now that we have introduced the possible client architectures within our scope, we start our comparison. Our analysis covers four categories:
\begin{enumerate}
    \item Security considerations: We discuss the security protocols that can be used for the different architectures, and what the implications of this are in terms of security guarantees and possible protocol overhead;
    \item Performance, rate: We derive the expected runtime for the protocols by analysing the loss terms, time per attempt and expected number of attempts needed to complete one round of the protocol;
    \item Performance, errors: We discuss and provide an overview of the error terms present for each architecture;
    \item Hardware cost and complexity: We discuss the complexity and cost of implementing each architecture and discuss design considerations.
\end{enumerate}
For the security and performance categories, we provide an overview in Tables~\ref{tab:security}, \ref{tab:rate} and \ref{tab:errors}, and the hardware cost and complexity is summarized at the end of Section~\ref{subsec:cost}.

\subsection{Security and protocol considerations}\label{subsec:security}

The security properties of the client architectures vary significantly in terms of proven guarantees, robustness to noise, and practical overhead requirements. There are four key aspects to consider: the existence and strength of composable security proofs, noise robustness properties, the generality of the protocol (i.e., for what types of computations it works), and the overhead required to achieve security guarantees. An overview of these properties is provided in Table~\ref{tab:security}, with a discussion followed below. While we cover the status of all these aspects for each client architecture, we note that this reflects only the current state of the protocols: further development may reduce overhead, introduce or strengthen security and noise-robustness guarantees, and improve generality.

\begin{table*}[t]
\caption{Security and protocol properties per client architecture.
  Filled circles~(\protect\symP) indicate the property holds;
  open circles~(\protect\symA) indicate it does not; blue-grey circles~(\protect\symC) indicate some caveat, see main text for details.
  Architectures that can access both BQP-only and universal protocols are listed on separate lines, as these come with different security trade-offs.
  Citations connected by~`+' must be composed to obtain the full security guarantee; citations separated by~`/' are alternatives. Overhead is given in terms of round repetitions $r$; minimal number of qubits to make the graph state corresponding to the computation $m$; number of dummy qubits $d$; number of trap qubits $t$; number of qubits in the gadget construction for WCP clients $k$; and numbers of qubits $h$, $g$ and $l$ specific to the protocols of Refs.~\cite{hayashi2015verifiable, takeuchi2018verification}\protect\newline
  $^\dagger$Stand-alone (non-composable) result with only inverse-polynomial security\protect\newline
  $^\ddagger$Security and verifiability claimed without complete formal proof.}
\label{tab:security}
\centering
\includegraphics{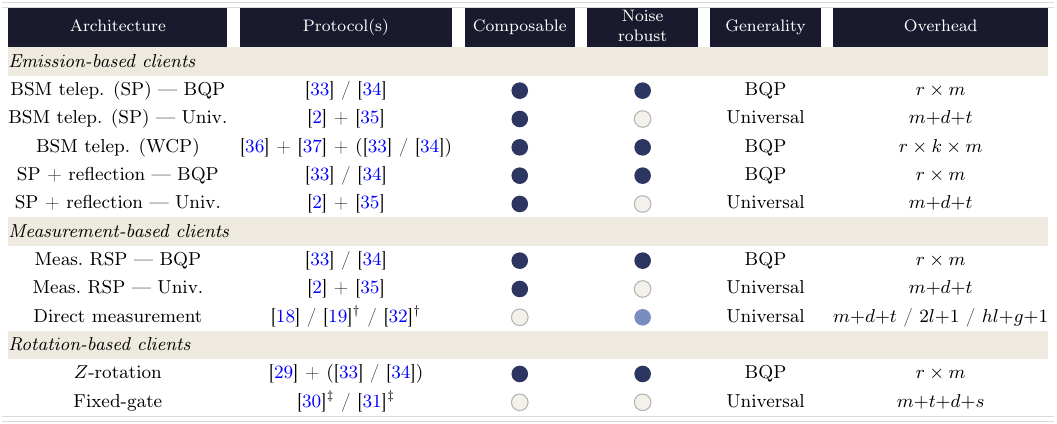}
\end{table*}

\textbf{Security proofs and reductions. } The distinction between prepare-and-send and receive-and-measure protocols, as introduced at the start of Section~\ref{sec:taxonomy}, is commonly used in security proofs. With this, there are generally two `flavours' of security proofs. Single-photon emission-based clients and measurement-based RSP clients fall under the prepare-and-send security flavour and thereby share security proofs \cite{leichtle2021verifying, kapourniotis2024unifying, fitzsimons2017unconditionally, dunjko2014composable}. The direct measurement client has the other flavour, though does not share this with other client types \cite{morimae2014verification, hayashi2015verifiable, takeuchi2018verification}. The security of the WCP-based client and the $Z$-rotation client are both based on reducing the setup to a prepare-and-send setup, allowing them to be plugged into existing prepare-and-send-type security proofs. For the fixed-gate clients, however, no such reduction proof exist and they thereby rely on their own stand-alone result.

The reduction for the WCP-based client is done in two steps, which are centred around dealing with the probabilistic mutli-photon events. First, Garnier et al. construct a BatchRSP resource from WCPs~\cite{garnier2024composably}, which guarantees that in a batch of $k$ transmitted states, with some exponentially high probability at least one is a genuine single-photon state. Then, this is combined with Ref.~\cite{kapourniotis2023asymmetric}, which provides a "collaborative RSP" gadget that distils BatchRSP into single photon RSP using a series of CNOT gates and measurements at the server, this can be plugged into existing prepare-and-send-type VBQC protocols that only require equatorial states ($|\pm_\theta\rangle$) (i.e., those that rely on stabilizer testing instead of trap-based testing).

\textbf{Composable security.} If security is proven in a composable way, this means it is proven to remain secure when used as a subroutine within larger quantum computations or when executed multiple times in sequence. This is a critical requirement for practical VBQC implementations, as these protocols are never deployed in isolation.

For each prepare-and-send architecture, a strong composable security proof is available, including the `extensions' to the $Z$-rotation client of Kashefi et al.~\cite{kashefi2024verification} and the WCP-based client~\cite{garnier2024composably, kapourniotis2023asymmetric}. The blindness of receive-and-measure protocols is based on the no-signalling principle, making it device-independent. However, the verifiability arguments in e.g.~\cite{morimae2014verification, hayashi2015verifiable, takeuchi2018verification} are not proven composable. Both Hayashi-Morimae~\cite{hayashi2015verifiable} and Takeuchi-Morimae~\cite{takeuchi2018verification} obtain only inverse-polynomial soundness.

The fixed-gate clients from Wu et al.~\cite{wu2023blind} and Li et al.~\cite{li2021blind} prove blindness in a non-composable way. In addition, Wu et al. claim verifiability but do not provide formal proofs. Moreover, there are potential vulnerabilities to correlated attacks. Unlike standard VBQC protocols that employ one-time pads to decorrelate successive measurements, the Wu et al.\ and Li et al. protocols lack this protection. A malicious server could potentially prepare entangled states instead of the prescribed single-qubit states and use auxiliary systems to exploit correlations across the multiple communication rounds. We note, however, that the $H/Z(\pi/4)$- and $H/X$-gate clients have the ability to perform bit flips (using $T^4=Z(\pi/4)^4=Z$ and $X=HZH$), so their security proofs might be adapted to mitigate this vulnerability.

\textbf{Noise robustness.} A critical security consideration in practical scenarios is noise robustness: whether the protocol aborts in the presence of honest errors. Without noise robustness, even protocols with fault tolerance become impractical, as they would constantly abort due to environmental decoherence, photon loss, or gate imperfections. The distinction between fault tolerance and noise robustness is crucial: a protocol may be fault-tolerant in principle (able to correct errors below some threshold) but not noise-robust (aborting frequently in realistic conditions), rendering it impractical for deployment.

Verification protocols that provide noise robustness are available for prepare-and-send type clients~\cite{leichtle2021verifying, kapourniotis2024unifying}, which therefore includes the WCP-based client and the $Z$-rotation client, which inherit noise robustness from being plugged into a noise-robust protocol. These protocols do not abort due to honest errors and can tolerate noise up to some threshold. In contrast, direct measurement architectures (receive-and-measure) and fixed-gate clients of Wu et al. and Li et al. lack noise-robust verification schemes in the current literature, making them not suitable for near-term implementation. Takeuchi and Morimae~\cite{takeuchi2018verification} explicitly discuss the impact of honest noise (their Sec.~VI) and show that their protocol still accepts slightly deviated states with high probability provided the deviation scales as $O(1/\mathrm{poly}(N))$, but they note that the protocol is not perfectly fault tolerant and that, for example, single-qubit phase flips already cause the verifier to reject with high probability.

\textbf{Protocol generality. } Not all protocols support the same classes of computations. Several of the protocols discussed above, specifically those in Refs.~\cite{leichtle2021verifying, kapourniotis2024unifying} and the first protocol of Ref.~\cite{kashefi2024verification} (Corollary 1),  and thus all noise robust protocols, are restricted to bounded-error quantum polynomial time (BQP) computations, meaning computations with classical input and output. This restriction is fundamental to their strategy: they repeat the computation over multiple rounds and take a majority vote to amplify correctness, a technique that requires the output to be classical.

The protocol in \cite{hayashi2015verifiable} only works on bipartite graphs, while Ref.~\cite{takeuchi2018verification} extends this to work on hypergraph states, though bipartite graphs (like the brickwork graph~\cite{broadbent2009universal}) are enough for universal computations. The prepare-and-send protocol in Ref.~\cite{dunjko2014composable} also works for any MBQC computation. The fixed-gate clients of Refs. \cite{wu2023blind, li2021blind} also claim universality, though without formal composable security proofs as discussed above. Ref. \cite{kashefi2024verification} additionally presents a protocol for arbitrary quantum computations (Corollary 2), but it requires the client to perform multi-qubit Clifford gates. This demands quantum memory and multi-qubit processing capabilities that are incompatible with the lightweight client architectures considered in this work, and we therefore exclude it from our comparison.

\textbf{Security overhead.} The security proofs for the different architectures realize their security in different ways, which may lead to different end-to-end runtimes of a verified blind calculation. While we discuss the performance of creating a single qubit of the computation graph separately in Sections~\ref{subsec:rate} and \ref{subsec:fidelity}, here we discuss performance considerations that arise from the corresponding security protocols. We first discuss a `base' amount of qubits, and then discuss how different protocols add on this. We discuss the specific overhead WCP clients introduce, and the different sizes of round-based and graph-size based overheads. 

As a base case, take the protocol's entangled resource state to be a universal graph large enough to cover the computation: a single graph of $m$ qubits. Protocols, on top of this minimal $m$ give an overhead in either graph size or number of rounds. Note that these concepts are not entirely separate: with a `measure as you go' approach, the graph state does not need to be created in full before the server can start making its measurements: If for a given qubit in the graph all of its neighbours are prepared or rotated, it can be entangled and measured. A graph size overhead therefore does not necessarily impose a requirement on the amount of qubits the server needs to hold simultaneously, and in the end also manifests as a time overhead, similar to the round overhead. The difference is that this graph-size-induced time overhead is likely super-linear, because each new qubit must be created within some cutoff time (set by the coherence time). This means generating the full graph will typically take multiple attempts, because it requires starting over whenever a qubit fails to be created within the cutoff. In that sense, a round overhead is more desirable compared to a graph size overhead, at least until the coherence time is long enough for the cutoff to not impact the completion of the graph significantly, i.e., in the limit of the cutoff time approaching infinity (or in the limit of no decoherence), the graph size overhead and round-based overhead have the same effect.

The only existing security proof for VBQC with WCP-based clients combines results from Garnier et al. and Kapourniotis et al.~\cite{garnier2024composably, kapourniotis2023asymmetric}, which requires a gadget construction to handle multi-photon events. For each node in the computation graph, the client must remotely prepare $k$ physical qubits (where $k$ depends on the desired security level and the mean photon number of the pulse), which the server then concatenates via CNOT gates and measurements. The optimal choice of $k$ represents a non-trivial optimization problem, balancing security guarantees against runtime and finding an optimal mean photon number. A lower mean photon number reduces the average number of multi-photon events and thus fewer qubits need to be concatenated, but also lowers the probability of emitting any photons and thereby successfully RSP'ing qubits. In any case, it requires multiple qubits to be remotely prepared for each node in the graph.

The noise robust protocols rely on classically combining the outcomes over multiple rounds \cite{leichtle2021verifying, kapourniotis2024unifying}. With this, they also separate their verification from their computation: some of the rounds are used entirely for computation, and some entirely for verification. Because of this, in each round the size of the graph can remain at the minimal value $m$. They require $r$ repetitions of a round with $m$ qubits per graph, coming to a total of $r\times m$ qubits, i.e., they have a \emph{round-based} overhead of $r$. Some of these $r$ rounds are used for verification, and some for computation. For the protocol of Leichtle~et~al., the amount of rounds depends on noise level in the system, the structure of the computation and desired level of security, this is studied in \cite{vandam2026optimizing}. For the protocol of Kapourniotis~et~al., this $r$ is strictly smaller, as it is made independent of the computation, and achieves a tighter security bound. When the WCP client is used in combination with these, the total amount of qubits comes to $r\times k \times m$. 

The receive-and-measure protocols of \cite{hayashi2015verifiable} and \cite{takeuchi2018verification} also have a \emph{round-based} overhead, where only one round is used for the computation, while all other rounds are used for verification. Hayashi-Morimae uses $2l+1$ copies of the graph state, where $l$ depends on the desired security level. These copies are partitioned into three groups: $l$ copies for testing one set of stabilizers ($X$-basis measurements on one bipartition, $Z$-basis on the other), $l$ copies of the complementary stabilizer test, and 1 copy for the computation. The factor of two comes from the bipartite structure of the graph state, which requires two complementary stabilizer tests. Takeuchi-Morimae uses $hl+g+1$ copies, where $l$ is again the number of tests per stabilizer, but we now need $h$ groups of this, given that it works on hypergraphs with $h$ stabilizer generators. The extra $g$ copies are discarded as part of the de Finetti reduction: starting from a permutation-invariant state, discarding these copies guarantees that the state on the remaining systems is approximately i.i.d..

Other protocols rely on a \emph{graph-size} overhead, where verification happens on a part of the graph. Specifically, the protocols in Refs.~\cite{fitzsimons2017unconditionally, morimae2014verification} require $t$ `trap' qubits and $d$ `dummy' qubits, where the dummy qubits are used to effectively decouple the trap qubits from the rest of the graph, and the traps to test the server through eigen-basis measurements. This leads to a total graph size of $m+d+t$ qubits. Both protocols described in Ref.~\cite{wu2023blind} (the $T$- and $T/X$-gate clients) and the protocol of Ref.~\cite{li2021blind} (the $H/Z(\pi/4)$-gate client) additionally require $s$ decoy qubits, so the required graph state has a total of $m+d+t+s$ qubits.

In general, we call the amount of qubits that are needed per graph for a given protocol $n$. E.g., for the Wu~et~al. protocol $n=m+d+t+s$, for the Hayashi-morimae protocol $n=2l+1$, and for the protocol of Leicthel~et~al. $n=m$. We'll use this $n$ to determine the rate per client type in the next section, which is then independent of the exact security protocol that is used. 

\subsection{Rate}\label{subsec:rate}

To assess the performance of each architecture, we first discuss the aspects that go into determining the rate at which each client architecture can perform VBQC. After, in Section \ref{subsec:fidelity}, we cover error terms. We discuss the rate and error terms separately for the sake of providing a clear overview, though they are not independent. Some protocols exhibit rate-fidelity trade-offs, such as the detection window for efficiency vs.\ dark counts. Extra errors may also lead to a round overhead in the protocol~\cite{leichtle2021verifying, kapourniotis2024unifying}. The security overhead considerations discussed in Section~\ref{subsec:security} also determine the overall performance. The performance of a given architecture thus depends on a combination of its physical performance and the overhead the corresponding security protocol induces. Here, we discuss only the considerations in physical performance. We first describe the photon losses that occur during the client-server interaction for each architecture, from which we deduce a probability of this interaction succeeding. We then discuss how long such an interaction takes. After, we combine these expressions for probability-per-attempt and time-per-attempt into an estimated time for finishing one round of the computation, taking into account a possible cutoff time at the server, limiting the time a qubit can spend in its memory. We underline the compounding penalty for rotation-based clients from multiple round trips and finally we discuss multiplexing strategies. We summarize these points in Table \ref{tab:rate}.

\begin{table*}[t]
\caption{Rate properties per client architecture: time per attempt ($t$),
  success probability per attempt ($p$), and end-to-end scaling per computation round for a graph of size $n$, subject to a cutoff of $c$ attempts.
  The prepare-and-send architectures, except for the measurement-based RSP architecture, have a different time-per-attempt depending on the location of heralding.
  End-to-end scaling for prepare-and-send architectures is shown for linear graph states and is derived in Appendix~\ref{app:rate_deriv} along with bounds for scaling in a brickwork graph.
  $^\dagger$ For full graph instead of per qubit, see main text.}
\label{tab:rate}
\centering
\includegraphics{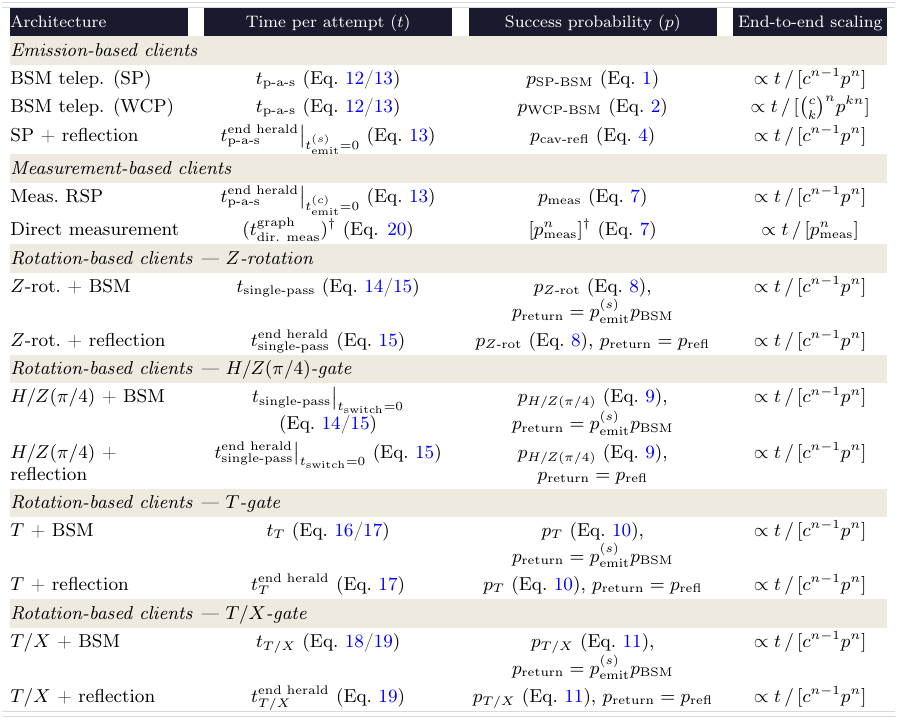}
\end{table*}

\textbf{Losses.} For all client architectures, there is some sort of interaction between the client and server through photonic qubits, for each node in the graph state. Each attempt at the interaction between the client and server is subject to several sources of loss, which we discuss below.

All architectures share a common contribution from transmission, which we will assume is done over fibre, as the distances are assumed to be in the order of tens of kms, or `metropolitan scale'. The probability of a photon surviving a channel of length $L$ is $\eta_\text{trans}(L) = 10^{-\alpha L/10}$, where $\alpha$ denotes the fibre loss in dB/km (typically 0.2~dB/km at telecom wavelengths). WCP-based clients can partially compensate for channel loss by increasing the mean photon number per pulse, though this comes at the cost of increased multi-photon events and a correspondingly larger gadget construction overhead (see Section~\ref{subsec:security}). Though there are differences between the architectures in how far the photon has to travel. Beyond transmission, each architecture introduces its own additional losses. We assume that the fixed-gate architectures use the static gate with looping method, and the other client architectures use dynamical controllers. Assuming that linear-optical manipulation losses are negligible, the dominant contributions for each architecture are as follows. 

For single-photon-BSM-based RSP, the success probability per attempt is
\begin{equation}
    p_\text{SP-BSM} = p_\text{emit}^{(s)} \cdot p_\text{emit}^{(c)} \cdot \eta_\text{trans}(L) \cdot p_\text{BSM} \,,
    \label{eq:p_sp_bsm_rsp}
\end{equation}
where $p_\text{emit}^{(s)}$ and $p_\text{emit}^{(c)}$ denote the end-to-end emission efficiencies of the server and client respectively (including coupling to fibre and any frequency conversion), and $p_\text{BSM}$ captures the Bell state measurement success probability. With linear optics, $p_\text{BSM} \leq 1/2$, but in practice this factor also absorbs detector inefficiencies and, crucially, losses arising from temporal filtering.

When the client replaces its single-photon source with a weak coherent pulse (WCP) of mean photon number $\mu = |\alpha|^2$, the budget equation acquires two modifications. First, the client emission efficiency $p_\text{emit}^{(c)}$ now refers to the coupling and frequency-conversion efficiency rather than a source heralding probability, and the mean photon number of the pulse $\mu$ enters as an independent parameter. Second, multi-photon contributions modify the success probability beyond the leading-order term. The result is derived in~\cite{vandam2025single}:
\begin{equation}
    p_\text{WCP-BSM} = p_\text{emit}^{(c)} \cdot \mu \cdot p_\text{emit}^{(s)} \cdot \eta_\text{trans}(L) \cdot p_\text{BSM} \cdot g\,,
    \label{eq:p_wcp_bsm_rsp}
\end{equation}
where $g$ is a multi-photon correction factor given by
\begin{equation}
\begin{split}
        g = &\frac{8\,e^{-\eta_c\mu/2}\bigl(1 - e^{-\eta_c\mu/4}\bigr)}{\eta_c\eta_s\mu} \, \times \\&\left[\frac{\eta_s}{2}\!\left(1+\frac{\eta_c\mu}{4}\right) + \bigl(1-e^{-\eta_c\mu/4}\bigr)(1-\eta_s)\right]\,,
\end{split}    
    \label{eq:g_multiphoton}
\end{equation}
with $\eta_s$ and $\eta_c$ the total path efficiencies of the server and client photons up to the BSM (detector efficiencies and temporal filtering are already accounted for in $p_\text{BSM}$). Note that in principle this can be expressed in terms of $p_\text{emit}^{(c)}$, $p_\text{emit}^{(s)}$ and $\eta_\text{trans}(L)$, but this is left ambiguous because the location of the BSM station (often assumed to be half-way between the client and server), and therefore what fraction of $\eta_\text{trans}(L)$ is taken as part of $\eta_c$ and $\eta_s$ is again a trade-off that can be optimized: putting the BSM closer to the server means we can increase $\mu$ to compensate for the extra losses on the client's side, but this does raise security overhead discussed in Section~\ref{subsec:security}. 

In the limit $\mu \to 0$, we have $g \to 1$, recovering the single-photon result. Increasing $\mu$ without introducing additional losses raises the success probability through both the explicit linear scaling and $g > 1$, but this comes at the cost of reduced fidelity of the remotely prepared state due to multi-photon noise; we quantify this trade-off in Section~\ref{subsec:fidelity}.

In a BSM, timing mismatches between the photons from the client and server can reduce the visibility, even if the photons are otherwise identical. One cause of timing mismatch, besides alignment and time synchronisation imperfections, is spontaneous emission of a different photon during the emission process, as in e.g. \cite{krutyanskiy2023entanglement}. The spontaneously emitted photon will not cause much trouble otherwise, as it will be so dissimilar to the targeted photons that the setup will be largely insensitive to it. It does, however, cause a delay in the targeted photon emission. Krutyanskiy et al. mitigate this by truncating the coincidence window at the BSM, only allowing detector clicks within a small time frame. This creates an explicit rate-fidelity trade-off: Truncating the detection window more means fewer distinguishability errors, but also results in a lower success rate. This rate-fidelity trade-off is encoded in $p_\text{BSM}$. 

For cavity-reflection-based interaction (encompassing SPRINT~\cite{rosenblum2017analysis} and CAPS-type gates~\cite{duan2004scalable, kikura2025passive}), the client emits a photon that travels to the server, where it reflects off the cavity containing the server's matter qubit. The reflected photon is detected to herald a successful gate operation:
\begin{equation}\label{eq:p_cav_refl}
p_\text{cav-refl} = p_\text{emit}^{(c)} \cdot \eta_\text{trans}(L) \cdot \eta_\text{coupl} \cdot p_\text{gate}(C_i) \cdot p_\text{detect}^{(s)}
\end{equation}
where $\eta_\text{coupl}$ is the fibre-to-cavity spatial coupling efficiency, $p_\text{detect}^{(s)}$ the servers detection efficiency, and $p_\text{gate}(C_i)$ is the gate success probability, which depends on the intrinsic cooperativity $C_i$ and approaches unity only in the limit $C_i \to \infty$.

For SPRINT~\cite{rosenblum2017analysis}, $p_\text{gate} = R(C_i)$, the reflection probability from Eq.~(7b) of~\cite{rosenblum2017analysis}, evaluated at optimal external coupling $\kappa_\text{ex}^\text{opt} = \kappa_i\sqrt{1+2C_i}$~\cite{rosenblum2017analysis}.

For CAPS-type gates~\cite{kikura2025passive}, two operating regimes exist, corresponding to a rate-fidelity trade-off. In the \emph{standard} regime, all reflected photons are accepted, giving
\begin{equation}
p_\text{CAPS} = 1 - \frac{\sqrt{1+2C_i}}{1 + C_i + \sqrt{1+2C_i}}
\end{equation}
with a residual conditional infidelity scaling as $(1+C_i)^{-1}$~\cite{kikura2025passive}. In the \emph{high-fidelity} regime, a calibrated attenuator is inserted in the cavity's off-resonant reflection path to equalise the two conditional reflection amplitudes, eliminating the amplitude mismatch that causes infidelity. This yields unit conditional fidelity at the cost of a lower success probability~\cite{kikura2025passive} and increased complexity:
\begin{equation}
p_\text{CAPS}^\text{opt} = 2\,p_\text{CAPS} - 1
\end{equation}
At $C_i = 100$, for example, these evaluate to $p_\text{CAPS} \approx 0.87$ and $p_\text{CAPS}^\text{opt} \approx 0.75$ respectively.

For measurement-based RSP, the server emits a memory-entangled photon that travels to the client for detection:
\begin{equation}
    p_\text{meas} = p_\text{emit}^{(s)} \cdot \eta_\text{trans}(L) \cdot p_\text{detect}^{(c)} \,,
    \label{eq:p_meas_rsp}
\end{equation}
where $p_\text{detect}^{(c)}$ is the client's detection efficiency. Since the client performs a projective measurement on a single incoming photon rather than interfering two photons, this architecture is insensitive to the temporal profile of the server's emission. Delayed emission events that would degrade BSM visibility are here simply detected and projected onto the correct basis, provided the photon arrives within the detector's acceptance window.

In the direct measurement approach, each qubit of one graph needs to reach the client, meaning the losses compound with the size of the graph $n$ as $p^n$, while for the other architectures the scaling is more favourable, assuming the coherence time is long enough. The per-qubit loss terms in the direct measurement approach are the same as for the measurement-based RSP approach, but the scaling is different: all qubits in an RSP-approach are independent, thus, you can try to remotely prepare each qubit as often as you want (apart from the problem of decoherence). This is revisited later in this section.  

For the rotation-based architectures, the photon must make one or more round-trips between client and server, and must be re-absorbed into the server's memory after each round-trip. The effective channel length therefore increases to $2L$, $4L$, and $6L$ for the single-pass, $T$-gate, and $T/X$-gate clients, respectively. The success probabilities are
\begin{align}
    p_{Z\text{-rot}} &= p_\text{emit}^{(s)} \cdot \eta_\text{trans}(2L) \cdot p_\text{return} \,, \label{eq:p_zrot} \\
    p_{H/Z(\pi/4)}&= p_\text{emit}^{(s)} \cdot \eta_\text{trans}(2L) \cdot p_\text{switch} \cdot p_\text{return}\,, \label{eq:p_HT}\\
    p_{T} &= \left[p_\text{emit}^{(s)}\right]^2 \cdot \left[p_\text{switch}^{(c)}\right]^2 \cdot \eta_\text{trans}(4L) \cdot \left[p_\text{return}\right]^2 \,, \label{eq:p_T} \\
    p_{T/X} &= \left[p_\text{emit}^{(s)}\right]^3 \cdot \left[p_\text{switch}^{(c)}\right]^3 \cdot \eta_\text{trans}(6L) \cdot \left[p_\text{return}\right]^3 \,, \label{eq:p_TX}
\end{align}
where $p_\text{switch}^{(c)}$ accounts for optical switching losses at the client (routing between the gate loop and the return path or delay line), and $p_\text{return}$ denotes the probability of successfully transferring the returning photon back onto the server's memory qubit. This transfer can be realised either through a BSM with a freshly emitted server photon, contributing $p_\text{return} = p_\text{emit}^{(s)} \cdot p_\text{BSM}$, or through a reflection-based SWAP or teleportation, in which case $p_\text{return} = p_\text{refl}$. The multiplicative structure of these expressions means that each per-round-trip loss factor enters with an exponent equal to the number of round-trips, compounding on top of the increased transmission loss.

Comparing across architectures, the measurement-based clients (both RSP and direct measurement) have the most favourable loss structure: they involve a single emission, a single traversal of the channel, and a single detection event, with no photon-photon interference requirements. The BSM-based emission clients add a second emission process and the BSM efficiency, including the temporal filtering overhead. The reflection-based client avoids the BSM but introduces the cavity-mediated SWAP probability. The rotation-based clients are the most loss-sensitive, as each loss channel is raised to a power equal to the number of round-trips, making their performance particularly dependent on the server--client distance and the quality of the photon-matter interface at the server.

\textbf{Attempt duration.} The losses discussed above determine how many attempts are needed on average for a successful state preparation. Here we discuss how long each attempt takes for the different architectures. We note that in these attempts, we do not take into account measuring the qubits. In a receive-and-measure approach, a measurement is assumed to be instantaneous, as it is a detection of photons. In other architectures, the time to read out the state of a qubit at the server might be significant compared to other timescales, but it only occurs after a success. Since we likely need many attempts to get a success, the readout time is expected to be insignificant compared to the time it takes to get a success. If this is not the case, because the efficiency is very high, this measurement time could be approximated in the time-per-attempt by simply adding a term $p\cdot t_\text{readout}$ that is an average readout time $t_\text{readout}$ per attempt. Similarly, we do not include the time it takes to perform the entangling operations, as this also only occurs after a success. This is not true for the direct measurement approach, as the qubits are entangled before they are sent to the client, meaning this is a cost that occurs for each attempt. 

In the previous section, we discussed the challenge of fast control of state manipulation. This time can be significant, and sometimes forces the server and client to slow down their emission rates to account for the client's basis adjustment. We denote client switching time by $t_\text{switch}$.

For all prepare-and-send type clients, the duration of one attempt is determined by three factors: the emission duration (including procedural steps like cooling), the communication time of the quantum signal and heralding signal (with $t_\text{com} = L/c_\mathrm{fibre}$ the one-way communication time, where $c_\mathrm{fibre}\approx \SI{2e8}{\meter\per\second}$ is the speed of light in fibre), and the switching time of the client's photon rotation $t_\text{switch}$. If both the client and the server emit, the slowest emitter is leading in the time-per-attempt. It is assumed the client can switch their state manipulation device directly after each emission, such that there is a full attempt time for the client to do so. Only if the client cannot adjust their state manipulator within such an attempt, the attempt rate has to slow down to accommodate $t_\text{switch}$.

The communication time (both the signal travel distance and the heralding communication time) differs between the rotation-based clients and the other prepare-and-send clients. The communication time further depends on the location of heralding: when using a BSM to teleport the client's state into the server, this BSM can be located anywhere between the client and the server. Most commonly, for an emitting type client, it will be located midway between the client and server. However, if it is located near the server, if a reflection-based approach is used, or if the heralding signal comes from the client, as in the measurement-based RSP client, the signals need to travel further. In these setups, extra communication time is added to each attempt, this is depicted in Figure \ref{fig:heralding}. 
\begin{figure}[htbp]
    \centering
    \begin{subfigure}{\columnwidth}
        \centering
        \includegraphics[width=\linewidth]{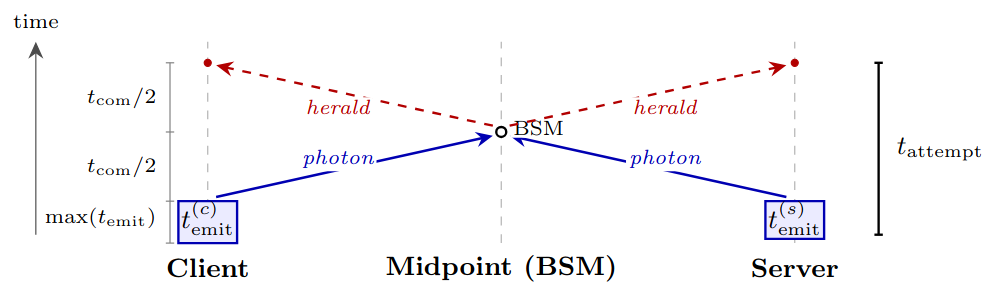}
        \caption{Midpoint heralding}
        \label{fig:midpoint_herald}
    \end{subfigure}
    
    \vspace{0.5cm} %
    
    \begin{subfigure}{\columnwidth}
        \centering
        \includegraphics[width=\linewidth]{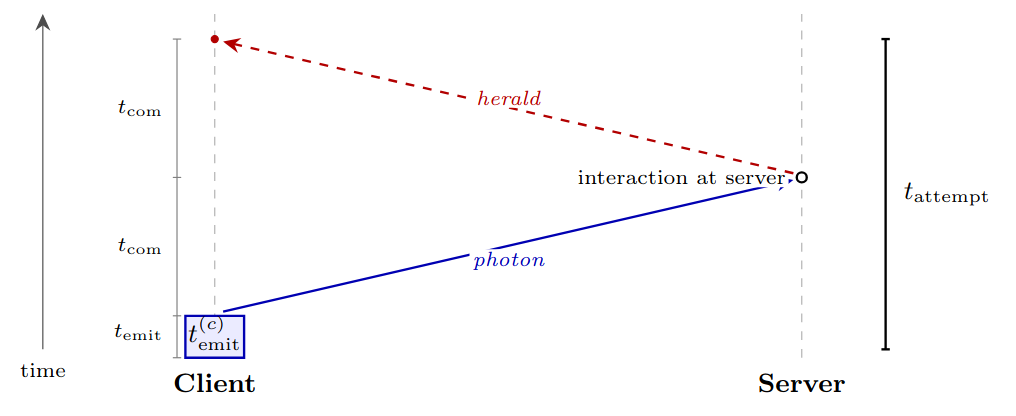}
        \caption{End-point heralding, depicted here for server-side heralding}
        \label{fig:end_herald}
    \end{subfigure}

    \caption{The time per attempt ($t_\text{attempt}$) depends on the location of the heralding. When heralding is done at the midpoint between the client and the server (a), less communication is required compared to the situation in (b), where heralding is done at an end-point.}
    \label{fig:heralding}
\end{figure}
For mid-point heralding, the total time of an attempt for prepare-and-send (p-a-s) type clients can be expressed as
\begin{equation}
    t_\text{p-a-s}^\text{mid herald} = \max\!\left[\max\!\left(t_\text{emit}^{(c)},\, t_\text{emit}^{(s)}\right) + t_\text{com},\, t_\text{switch}\right].
    \label{eq:t_pas_mid}
\end{equation}
Whereas for a server-side or client-side herald, it is given by
\begin{equation}
    t_\text{p-a-s}^\text{end herald} = \max\!\left[\max\!\left(t_\text{emit}^{(c)},\, t_\text{emit}^{(s)}\right) + 2t_\text{com},\, t_\text{switch}\right].
    \label{eq:t_pas_server}
\end{equation}
For the measurement-based RSP and single-pass manipulation clients there is no client emission time to take into account, so we can set $t_\text{emit}^{(c)}=0$. For a WCP-based client, the clients emission time is unlikely to be limiting, given high rates of WCP sources. Additionally, for the single-pass manipulation clients there is an additional communication leg, 
\begin{align}
    t_\text{single-pass}^\text{mid herald} =& \max\!\left[t_\text{emit}^{(s)} + 2t_\text{com},\, t_\text{switch}\right],\label{eq:t_single_mid}\\
    t_\text{single-pass}^\text{end herald} =& \max\!\left[t_\text{emit}^{(s)} + 3t_\text{com},\, t_\text{switch}\right].
    \label{eq:t_single_serv}
\end{align}
With $t_\text{switch}=0$ for the $H/Z(\pi/4)$-gate client.

The $T$- and $T/X$-gate clients require multiple round-trips. Here again, the time-per-round for these multi-pass clients depends on the location of the heralding. Heralding at the server takes an additional $t_\text{com}$ per round-trip compared to heralding at a midpoint station. We then have
\begin{align}
    t_{T}^\text{mid herald} &= 2t_\text{emit}^{(s)} + 4t_\text{com} \,, \label{eq:t_T_mid} \\
    t_{T}^\text{end herald} &= 2t_\text{emit}^{(s)} + 6t_\text{com} \,, \label{eq:t_T_serv} \\
    t_{T/X}^\text{mid herald} &= 3t_\text{emit}^{(s)} + 6t_\text{com} \,, \label{eq:t_TX_mid}\\
    t_{T/X}^\text{end herald} &= 3t_\text{emit}^{(s)} + 9t_\text{com} \,. \label{eq:t_TX_serv}
\end{align}

For the direct measurement approach, the concept of time per attempt becomes a bit different. For all other architectures, we have looked at the time and probability of remotely preparing or manipulating one qubit. There, one re-tries that one qubit until success, or until another qubit in the graph reaches its cutoff time. In the direct measurement approach, the photon that is sent to the client is already entangled with the rest of the register, losing it therefore means the whole graph needs to be built again. Because of this, the client also does not have to classically communicate a success/failure message for each qubit back to the server. Instead, the server can send over the full graph state, and the client only has to communicate whether or not the full graph state arrived. The time \emph{per qubit} only depends on either how fast the server can initialize and emit its qubits, or on how fast the client can switch its measurement basis. We then define the time of one \emph{attempt} for an $n$-qubit graph as one attempt of sending over the full graph, meaning we first initialize the graph $t_\text{init-graph}$, which includes entangling operations, then perform the \emph{per qubit} operations $n$ times and finish off with waiting for the qubits to arrive and receiving the heralding signal.
\begin{equation}\label{eq:tgraph_meas}
    t_\text{dir. meas}^\text{graph}= t_\text{init-graph} + n\cdot \max\!\left(t_\text{emit}^{(s)}, t_\text{switch}\right) + 2t_\text{com}\,.
\end{equation}
We note that the server does not need to prepare the full graph before it starts sending qubits to the client, instead, it can send any qubit to the client that is already fully entangled in the graph. Still, this equation holds under two assumptions: First, the server cannot perform emissions and graph-preparation operations in parallel. Second, the server cannot abort on a loss herald from the client. This second assumption is true when a loss herald from the client does not reach the server in time, i.e., the preparation of the graph is faster than the round-trip communication time. 

\textbf{Round-trip penalty for rotation-based clients.} The rotation-based client architectures suffer a compounding performance penalty from the multiple communication legs required per attempt, both through increased transmission loss and through longer attempt durations. To quantify this combined effect, we compare each architecture against the emission- and measurement-based clients as a baseline.

The expected time per successful state preparation scales as $t_\text{attempt}/p_\text{success}$. Considering transmission as the sole source of loss, the slowdown ratio for a rotation-based architecture $X$ relative to the emission- and measurement-based baseline is
\begin{equation}
    R_X = \frac{t_\text{emit} + \ell_t\, t_\text{com}}{t_\text{emit} + 2\, t_\text{com}} \cdot \frac{\eta_\text{trans}(L)}{\eta_\text{trans}(\ell_q L)} \,,
    \label{eq:slowdown}
\end{equation}
where $\ell_q$ is the number of quantum communication legs relevant for loss ($\ell_q = 2$, $4$, and $6$ for the single-pass manipulation clients, $T$-gate, and $T/X$-gate clients, respectively), and $\ell_t$ is the total number of communication legs relevant for attempt duration ($\ell_t = 3$, $6$, and $9$, respectively). In the regime where $t_\text{emit} \ll t_\text{com}$, this simplifies to
\begin{equation}
    R_X \approx \frac{\ell_t}{2} \cdot 10^{\alpha(\ell_q - 1)L/10} \,.
    \label{eq:slowdown_simplified}
\end{equation}
To illustrate the magnitude of this penalty, consider $L = 25$~km and $\alpha = 0.2$~dB/km. The one-way communication delay is $t_\text{com} = L/c_\mathrm{fibre} \approx 125~\mu$s, which typically dominates over $t_\text{emit}$. At this distance, the emission- and measurement-based clients have approximate single-attempt transmission probabilities of about 0.3, while the single-pass manipulations clients, $T$-gate, and $T/X$-gate clients achieve approximately 0.1, 0.01, and 0.001, respectively. Applying Eq.~\eqref{eq:slowdown_simplified}, this gives slowdown factors of approximately $5\times$, $95\times$, and $1423\times$ for the single-pass manipulation clients, $T$-gate, and $T/X$-gate architectures, respectively. We stress that this estimate accounts only for transmission loss; including the architecture-dependent loss terms from Equations~\eqref{eq:p_zrot}--\eqref{eq:p_TX}, where each additional loss channel is traversed multiple times, will further widen these gaps.

\textbf{Total round duration. } Up until now we discussed the probability of one client-server interaction succeeding, and thereby we know how many attempts on average it takes to get one successful interaction. We also discussed how long one such interaction takes, and, in Section \ref{subsec:security}, how many qubits are needed for one round of the protocol. Here, we will work on putting those things together into one metric: the expected time to finish one round of the protocol. To disentangle the specific security proof from the client architecture, we will work with a general number of qubits per graph $n$, which can then be filled-in with the number of qubits needed for a specific protocol. The time to successfully create the whole graph state, and thereby complete one round of the protocol, is a random variable $T$. Since the duration $t$ of a single attempt depends on the architecture, we introduce the random variable $K$ which gives the number of attempts needed to complete one round. The expected time to create the graph state can be obtained as $\EE(T)=\EE(K)t$.

For the direct measurement approach, a limiting factor is that the qubits are entangled before transmission. This means that if any of the qubits is lost, the whole graph needs to be recreated and sent anew. One attempt in this architecture therefore corresponds to sending the full graph state of $n$ qubits. As the probability of one qubit being detected at the client is given by $p_\text{meas}$ (Equation \ref{eq:p_meas_rsp}), the probability of the whole graph arriving is $p_\text{meas}^n$. The client and server thus will take on average $\EE(K_\mathrm{meas})=1/p_\text{meas}^n$ attempts to send the whole graph, where each attempt takes $t_\text{dir. meas}^\text{graph}$ (Equation \ref{eq:tgraph_meas}). The expected time to complete one round in the direct measurement approach is then
\begin{equation}
    \EE(T_\text{meas}) = \frac{t_\text{dir. meas}^\text{graph}}{p_\text{meas}^n} \,.
\end{equation}

For all other architectures the scaling with the success probability of the single-qubit interaction can be better. We assume that the client and server prepare the graph state using the \textit{measure-as-you-go strategy}, in which a qubit in the graph state is measured as soon as all its neighbours have been generated and the relevant $CZ$-gates have been performed. The qubits are prepared sequentially according to a predetermined flow that specifies the order in which the qubits must be generated. Each attempt now corresponds to generating a single qubit of the graph state, and succeeds with probability $p$. To limit decoherence, a cutoff is imposed on the number of attempts a qubit may wait in memory before it is measured. We consider \emph{qubitwise cutoffs} $c$ in which a qubit must be discarded if it cannot be measured within $c$ attempts after being first generated. If a qubit in the graph violates this cutoff condition, then the whole graph has to be discarded and the generation has to start from scratch.

The cutoff therefore induces a trade-off between the decoherence accumulated in the graph state and the time it takes to create it. In \cref{app:rate_deriv} we study the expected number of attempts $\EE(K)$ to generate two specific graphs: a simple linear graph, which can be used for single-qubit computations~\cite{danos2007measurement}; and a brickwork graph, which is a universal resource for MBQC~\cite{broadbent2009universal}. The derivations use Wald's equation applied to a renewal process with a stopping condition (Theorem 3.3.2 in Ref. \cite{ross1995stochastic}). The results are discussed below. 

For a linear graph of $n$ qubits we denote the number of attempts to create the full graph state in the measure-as-you-go paradigm subject to a qubitwise cutoff $c$ by $K_\mathrm{linear}(c)$. Its expectation values is (\cref{eq: expected number of attempts linear graph})
\begin{equation}\label{eq: main: expected number of attempts linear graph}
    \EE\big(K_\mathrm{linear}(c)\big) = \frac{\frac{1}{p} +\frac{1}{p}\sum_{j=1}^{n-1}(1-(1-p)^c)^{j}}{(1-(1-p)^c)^{n-1}}.
\end{equation}
If the coherence time is small compared to the expected time for one successful qubit, then the cutoff must satisfy $c\ll 1/p$ to ensure a high quality graph state. In this regime, the expected time to complete one round scales as
\begin{equation}
    \EE\big(T_\mathrm{linear}(c)\big) \sim \frac{t}{c^{n-1}p^{n}}.
\end{equation}
For a cutoff $c=1$ this yields the same scaling with $p$ as the direct measurement approach, but the scaling can be improved by increasing the cutoff. In fact, if the coherence time is large compared to the expected time for one successful qubit, then the cutoff can also be large. In the limit of $c\to \infty$ this leads to the expected time to create the graph equal to $n$ times the expected time to prepare one qubit:
\begin{equation}
    \lim_{c\to \infty} \EE\big(T_{\text{linear}}(c)\big) = \frac{nt}{p}.
\end{equation}
The $1/p$ scaling in this regime is much better than the $1/p^n$ scaling of the direct measurement approach.

For a brickwork graph with $n_\mathrm{r}$ rows and $n_\mathrm{c}$ columns, for a total of $n=n_\mathrm{r}n_\mathrm{c}$ qubits, we consider a measure-as-you-
go flow in which the qubits are generated from
top to bottom in each column, and the columns are generated
from left to right. The qubitwise cutoff $c$ then imposes the constraint that a qubit in row $i$ and column $j$ must be generated within $c$ attempts of generating the qubit in row $i$ and column $j-1$. If this does not happen, the latter qubit cannot be measured within $c$ attempts after being created, and violates the cutoff condition.

It is difficult to obtain the scaling of the expected time $\EE\big(K_\mathrm{brickwork}(c)\big)$ directly because the qubitwise cutoffs lead to many interdependent conditions on the number of attempts allowed for each qubit (c.f. \cref{eq: condition succ B}). However, the expected number of attempts subject to qubitwise cutoffs can be bounded by considering the coarser \emph{columnwise cutoffs}, in which every column of the graph state must be generated within $c$ attempts.
Intuitively, a columnwise cutoff $c$ is a looser condition than a qubitwise cutoff $c$. The columnwise cutoff $c$ only places the constraint that the first qubit of a column is generated within $c$ attempts of the first qubit of the previous column, while the qubitwise cutoff places this constraint on all qubits in the column. On the other hand, a columnwise cutoff of $c/2$ makes sure that two consecutive columns are generated within $c$ attempts, so that all qubitwise cutoffs between these columns are definitely satisfied. The precise formulation of these intuitions can be found in \cref{lemma: succesful path inclusion lemma}, and leads to the bounds of \cref{cor: Bounds on expected time to create full graph}. 

Assuming $c\gg n_\mathrm{r}$, i.e. that the cutoff is much larger than the number of rows, the scaling of the expected time with qubitwise cutoffs to lowest order in $p$ is
bounded as (c.f. \cref{cor: Bounds on expected time to create full graph} and \cref{eq: lowest order p scaling columnwise cutoff brickwork})
\begin{equation}\label{eq: brickwork scaling column}
    \frac{t}{\binom{c}{n_\mathrm{r}}^{n_\mathrm{c}}p^{n}}\lesssim \EE\Big(T_\mathrm{brickwork}(c)\Big) \lesssim \frac{t}{\binom{c/2}{n_\mathrm{r}}^{n_\mathrm{c}}p^{n}}.
\end{equation}
Again, the scaling with $p$ can be much better than for the direct measurement approach, depending on the cutoff. 
We show the scaling with $p\in [10^{-3},10^{-1}]$ of the upper and lower bounds of \cref{cor: Bounds on expected time to create full graph} on the expected number of attempts for a brickwork graph state with $n_\mathrm{r}=2$, $n_\mathrm{c}=10$ and qubitwise cutoff $c$ in \cref{fig:nr2nc10c1000}. These numbers correspond for example to an attempt duration $t\sim \SI{100}{\micro\second}$ and coherence time of $\sim \SI{1}{\second}$, which is $10^4\,\mathrm{attempts}$, where the cutoff on the number of attempts is chosen an order of magnitude below the coherence time of $10^4$ attempts. In this example, the upper and lower bound are tight for $p\gtrsim 10^{-2}$. This is the regime where the graph succeeds with very little to no resets (i.e., we rarely violate the cutoff condition) as the expected time $n_\mathrm{r}/p$ to prepare one column of qubits is close to or below the cutoff $c$. In other words, the cutoff condition can be satisfied 
often if $n_\mathrm{r}/p\ll c$, that is, when $p\gg n_\mathrm{r}/c$. In this case, $n_\mathrm{r}/c=2\cdot 10^{-3}$, and we see from the blow up of the expected number of attempts that for $p\lesssim 10^{-2}$ the cutoff condition is violated often and many resets are triggered. Because of the many resets, the generation time blows up. Similar behaviour is seen for other sizes of graph states where the lower and upper bound are tight when there are few resets, but blow up and are separated by a multiple orders of magnitude when resets become dominant, c.f. Figure \ref{fig:examples bounds}. However, the fact that both blow up show that graph state creation subject to qubitwise cutoffs quickly becomes infeasible when $p$ is not significantly larger than $n_\mathrm{r}/c$.

Finally, aside from its use in bounding the qubitwise cutoff for a brickwork graph state, the columnwise cutoff condition is a sensible condition for cutoffs in linear graph states of size $n$ of which each node consists of $k$ qubits, constructed with the BatchRSP resource of the WCP client \cite{garnier2024composably}: it requires that all qubits in the BatchRSP are generated within a cutoff $c$, so that the decoherence in each BatchRSP is limited separately. For a resource of size $k\ll c$ the scaling of the lower bound in Equation \ref{eq: brickwork scaling column}, which corresponds to the scaling for columnwise cutoffs $c$ in a brickwork graph, yields, after renaming the variables,
\begin{equation}
    \EE\big(T_\mathrm{linear-WCP}(c)\big) \sim  \frac{t}{\binom{c}{k}^{n}p^{kn}}.
\end{equation}

\begin{figure}
    \centering
    \includegraphics[width=\linewidth]{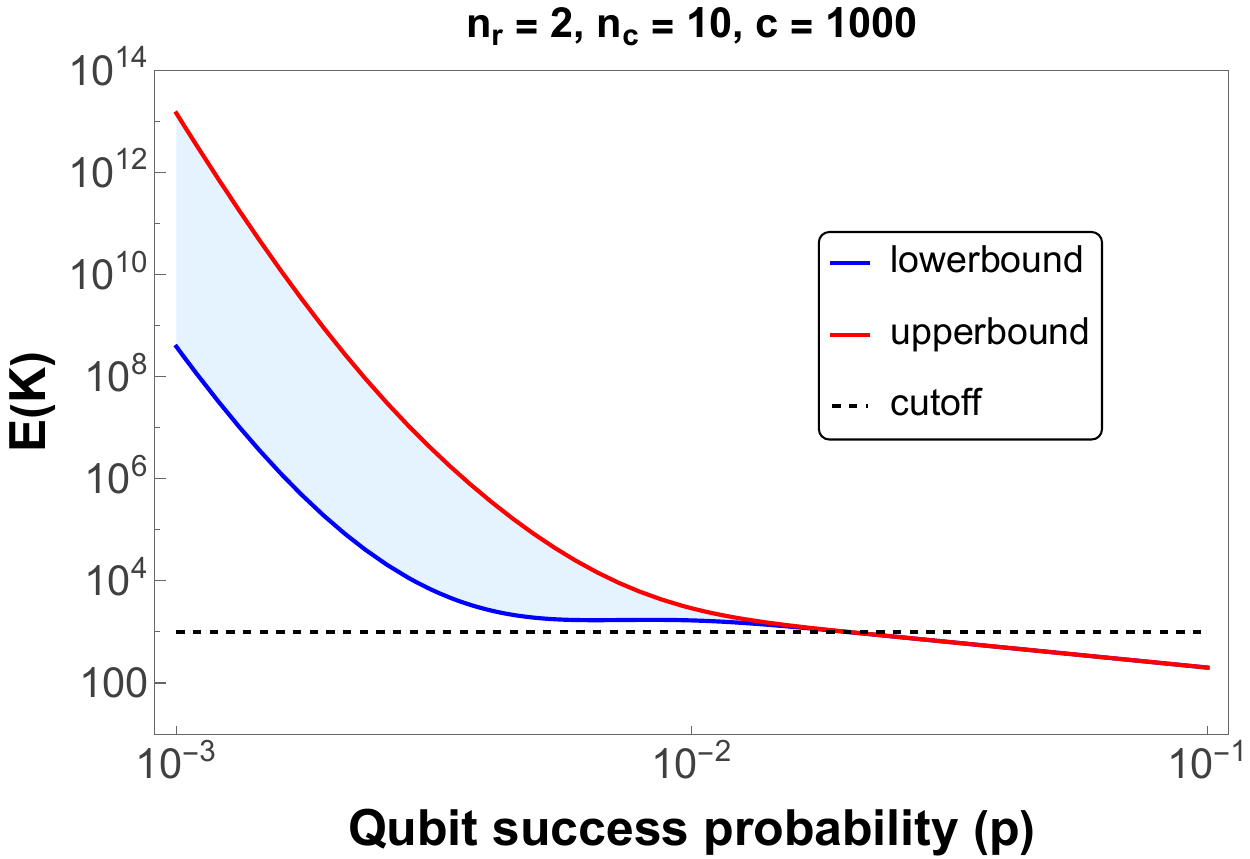}
    \caption{Upper and lower bounds to expected number of attempts for generating brickwork graph state with cutoffs. The upper and lower bounds on the expected number of attempts to generate a brickwork graph state with qubitwise cutoffs are evaluated numerically for a brickwork graph of dimensions $n_\mathrm{r}=2$ and $n_\mathrm{c}=10$ with cutoff $c=1000$ attempts in the regime $p\in [10^{-3},10^{-1}]$. The cutoff is shown as a horizontal dashed line. For $p\gtrsim 10^{-2}$ the expected number of attempts is less than or on the order of the cutoff. In this regime the upper- and lower bound are tight. For $p\lesssim 10^{-2}$ the upper bound blows up. The lower bound initially stays around the cutoff when decreasing $p$ below $10^{-2}$, but then blows up too. This bounds the range of the expected number of attempts within some orders of magnitude.}
    \label{fig:nr2nc10c1000}
\end{figure}

\textbf{Multiplexing.} For any practically useful application of blind quantum computing, multiplexing is crucial. Consider for example FeMoco, the canonical chemistry benchmark~\cite{reiher2017elucidating}. In the circuit model, this uses $n_\mathrm{r} = 100$ logical qubits and $10^6$ gate layers (trotterized simulation). In MBQC format, we need a brickwork graph state of size $n = n_\mathrm{r} \times n_\mathrm{c}$, with $n_\mathrm{r}$ the number of qubits in the circuit model and $n_\mathrm{c} = d \times g$ the number of columns, where $d$ is the circuit depth and $g$ is the gate decomposition overhead (5--15). For FeMoco this comes down to at least about $10^9$ qubits. They do not all have to exist simultaneously, but they must all be created or manipulated remotely by the client, making this the minimal number of photons that must travel between client and server, plus a heralding signal to return. Even if transmission were perfect, we never have to infer a cutoff, and all signals travel at the speed of light in vacuum $c_\mathrm{vac}$, this would still require $2L/c_\mathrm{vac} \times n$ time, amounting to roughly 46 hours at $L=25$~km. This is an absolute minimum, and in practice transmission is far from unity, and without $p\gg n_\mathrm{r}/c$ the expected time will blow up due to cutoff-induced resets. On top of this comes security overhead, and unless everything is perfectly noise free, the computation would have to be repeated to achieve statistical accuracy (this is built into the security overhead for noise-robust models). The only solution to overcome this impracticality is multiplexing.

The suitability of each client architecture for multiplexing varies considerably, driven primarily by three factors: the server qubit occupation time per attempt, the cost of duplicating client hardware, and the directionality of the quantum channel. We refer to~\cite{propp2025quantum} for a detailed treatment of both classical and quantum multiplexing strategies for quantum network protocols, including remote state preparation; here we discuss the qualitative differences between the client architectures.

For all RSP-based architectures (both emission- and measurement-based), each qubit preparation is independent, and the server qubit is occupied only for approximately $2t_\text{com}$ per attempt. This independence allows straightforward spatial multiplexing: multiple server qubits can attempt RSP in parallel without requiring coherence between them. Among these, WCP-based clients are particularly well-suited for multiplexing: a coherent state can be split over $S$ spatial modes using a beam splitter array at negligible cost, and the preparations require only a unidirectional quantum channel, simplifying wavelength- or time-division multiplexing over shared fibre infrastructure. Measurement-based RSP shares the advantage of unidirectional channels and independent preparations, though client duplication requires additional detectors. For single-photon emission clients, the duplication cost depends on the physical platform: in systems where multiple emitters can share an existing infrastructure (e.g., adding ions to a trap), the marginal cost of additional client emitters can be low, whereas platforms requiring independent source setups face the matching requirements for indistinguishability of the photons which we will discuss in Section~\ref{subsec:cost}. In all cases, the protocol structure for emission- and measurement-based clients is equally parallelisable.

The rotation-based architectures face a compounding multiplexing penalty. The server qubit occupation time per attempt increases to $3t_\text{com}$, $6t_\text{com}$, and $9t_\text{com}$ for the single-pass manipulation, $T$-gate, and $T/X$-gate clients respectively, while the success probability per attempt simultaneously decreases (Equations~\ref{eq:p_zrot}--\ref{eq:p_TX}). Since the number of server qubits required to sustain a given preparation rate scales as $t_\text{attempt}/p_\text{success}$, the server-side multiplexing resources grow rapidly. Furthermore, the photon must traverse the channel bidirectionally, complicating wavelength- or time-division multiplexing due to crosstalk between outgoing and returning photons across multiplexed channels.

The direct measurement approach presents a distinct set of considerations. Because the server must coherently emit the qubits of a graph state, spatial multiplexing requires the server to maintain entanglement across all simultaneously in-flight qubits, coupling the degree of parallelism directly to the server's quantum memory capacity and connectivity. This contrasts with the RSP-based architectures, where no such inter-qubit coherence is needed. On the other hand, the direct measurement client does not require heralding per qubit, as it only matters whether or not the full graph was transmitted successfully.

\subsection{Noise and errors}\label{subsec:fidelity}

After discussing considerations and scaling in rate, we now discuss the other side of the coin: errors. As stated before, rate and errors are often not independent, and we will underline whenever error terms are tightly connected to rate considerations. When there is a rate-fidelity trade-off, the optimal point is context dependent. Below, we'll go over all noise and error terms, and discuss how each of them affects the different architectures, which is summarized in \ref{tab:errors}. 

\begin{table*}[t]
\caption{Dominant error sources per client architecture.
  Filled circles~(\protect\symP) indicate the error source is present;
  open circles~(\protect\symA) indicate it is not.
  Blue-grey circles~(\protect\symC) indicate the error is present but
  with a relevant caveat, reduction or note (see text).
  Pills~(\protect\pillP{$n$}) indicate the error compounds $n$ times due
  to multiple round-trips or detection events; blue-grey pills~(\protect\pillC{$n$}) indicate
  compounding with a caveat. The decoherence is listed when the per-photon creation probability ($p$) is much smaller than the number of rows in the graph ($n_\mathrm{r}$) divided by the cutoff ($c$), i.e., for $p\ll n_\mathrm{r}/c$, for $p\gg n_\mathrm{r}/c$ the decoherence depends on the rate and flow of the graph (see Section~\ref{subsec:fidelity}). For the direct measurement client, the decoherence depends on the local qubit initialization or preparation time $t_\mathrm{init}$.}
\label{tab:errors}
\centering
\includegraphics{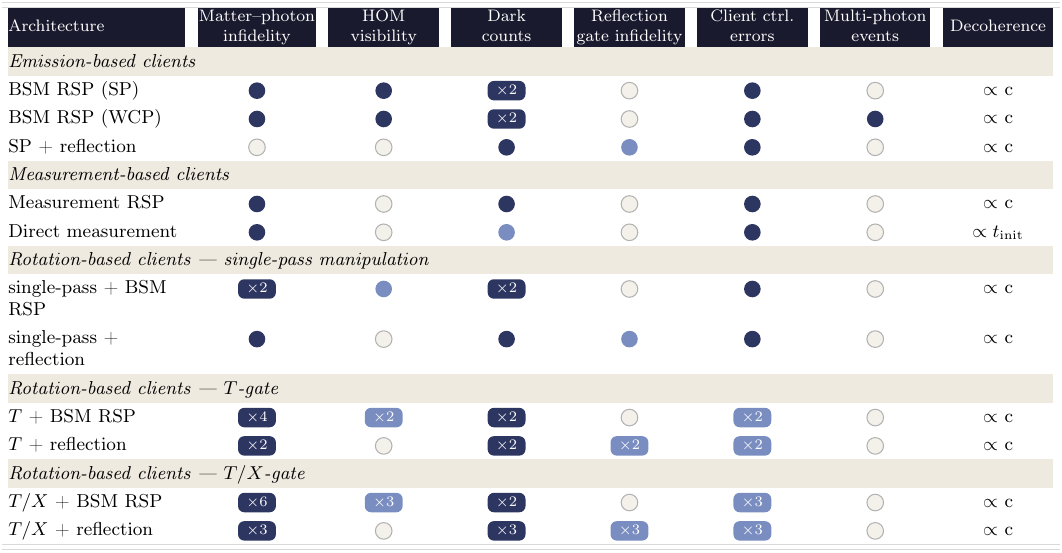}
\end{table*}

\textbf{Matter-photon infidelity. } Many architectures rely on the server emitting photons entangled with a matter-based memory qubit. In fact, this is true for all architectures that do not rely on cavity reflection techniques. In the emission-based and measurement-based RSP, infidelity in the matter-photon entangled state directly relates to infidelity in the remotely prepared state, and in the direct-measurement approach, it determines the fidelity of the graph state the client performs their measurement on. Similarly, for rotation-based clients, the server supplies the photons to the client, so infidelities here become infidelities in the remotely prepared state, which occurs multiple times per qubit: either once per round-trip if a reflection-based approach is used, or twice per round-trip for a BSM approach. 

\textbf{Hong-Ou-Mandel visibility. } The quality of the BSM depends on the indistinguishability of the photons that interfere in it. Considering distinguishability as the sole imperfection, the teleported-state fidelity is determined by the visibility $V$ as $F=(1+V)/2$. where $V=1$ for perfectly indistinguishable photons. For the rotation-based clients that use the BSM approach to transfer the state back onto the server, the photons will have the same source: they both originate from the server. While this does not guarantee perfect visibility, it greatly reduces errors compared to BSMs between dissimilar photon sources.

As a concrete example of the visibility effect with two different photon sources, assume a trapped-ion server interfacing with a quantum dot client emitter. The dominant source of distinguishability is the linewidth mismatch between the two emitters. For photons with Lorentzian spectra of linewidths $\Gamma_1$ and $\Gamma_2$, the spectral overlap, and thus the HOM visibility, is given by Ref.~\cite{legero2003time}
\begin{equation}\label{eq:visibility}
    V = \frac{4\Gamma_1\Gamma_2}{(\Gamma_1 + \Gamma_2)^2}.
\end{equation}
For a $^{174}$Yb$^+$ ion interacting via its 935\,nm transition, the atomic dipole decay rate is $\Gamma_\text{ion} \approx 2\pi \times 4.2$\,MHz, while the InAs quantum dot radiative linewidth is $\Gamma_\text{QD} \approx 2\pi \times 250$\,MHz, yielding a mismatch of approximately $60\times$~\cite{meyer2015direct}. Substituting into Equation~\ref{eq:visibility} gives $V \approx 0.065$, and therefore a prepared-state fidelity of only $F \approx 53\%$, barely above the classical limit of $50\%$. Mitigation strategies such as Purcell enhancement~\cite{purcell1946resonance} and spectral filtering can partially recover visibility, but these add experimental overhead, and residual mismatch remains a
significant source of infidelity.

Timing mismatches also decrease visibility, which can be mitigated by temporally filtering the photons, for example by imposing a strict time window in which successes are allowed at the BSM. This gives another rate-fidelity trade-off as reducing this window lowers rate, but improves overlap and therefore fidelity of the teleported state. 

\textbf{Dark counts. } Detector dark counts can cause false heralding. These can occur either at the BSM, further reducing the quality of the SWAP operation, or at the client in the measurement-based architectures. For symmetric threshold detectors (e.g., for polarisation-encoded photons, the probability for having a dark count in H is the same as for V) and a single-excitation dual-rail qubit, this noise is effectively depolarizing.

\textbf{Cavity-reflection gate infidelity.}
All cavity-reflection-based schemes share a common source of gate infidelity: amplitude mismatch between the two conditional reflection amplitudes. When the atom is decoupled from the cavity the photon reflects with amplitude $r_0 \approx +1$; when coupled, it reflects with amplitude $r_1$ that approaches $-1$ only as to cooperativity $C_i \to \infty$. At finite cooperativity, $|r_0| \neq |r_1|$, leaving a residual conditional infidelity. How each scheme handles this mismatch determines its error structure.

In the standard CAPS regime, all reflected photons are accepted and the residual infidelity scales as $(1+C_i)^{-1}$~\cite{kikura2025passive}. In the high-fidelity CAPS regime, a calibrated attenuator in the off-resonant reflection path equalises the two amplitudes, yielding unit conditional fidelity at the cost of a reduced success probability $p_\text{CAPS}^\text{opt} < p_\text{CAPS}$~\cite{kikura2025passive}. The gate infidelity is therefore absorbed into the rate rather than appearing as a separate error term.

For SPRINT, the optimal coupling condition $\kappa_\text{ex}^\text{opt} = \kappa_i\sqrt{1+2C_i}$ enforces complete forward destructive interference, so finite cooperativity and intrinsic cavity loss reduce the efficiency but not the conditional fidelity. Residual infidelity instead arises from imperfect circular polarization of the evanescent WGM field (parameterized by $r_\sigma$ and $r_\pi$), unwanted backscattering between counter-propagating modes at rate $h$, and interference between multiple excited-state pathways. When acting in isolation, $r_\sigma$ or $r_\pi$ caps the optimized fidelity at $F = 1/(1+2r_j^2)$ with $j \in \{\sigma,\pi\}$ (Eqs. 14, 18 of Ref.~\cite{rosenblum2017analysis}), while $h$ degrades $F$ quadratically (Eq. 23). The multi-excited-state effect can be fully compensated by a small cavity detuning. These imperfections are platform-dependent: for instance, a static magnetic field reduces $r_\sigma$ and $h$ and thus improves fidelity.

The gate infidelity entry for all cavity-reflection rows in Table~\ref{tab:errors} is marked \symC: infidelity is present in all variants but its magnitude and origin differ, and in the high-fidelity CAPS case it is better treated as a rate penalty instead.

Whether atoms left outside the qubit subspace by unwanted transitions contribute to infidelity or only to rate loss depends on whether such leakage events are heralded by the server's readout mechanism.

\textbf{Client control errors. } In all client architectures, the client uses linear optics to control photonic quantum states. Imperfections in this control hardware will effect all client architectures similarly, except for the fixed-gate clients. Here, the control hardware is not dynamically reconfigurable, and it can therefore be aligned more carefully, leading to fewer errors. There is, however, also a compounding effect: not only are there multiple round-trips past the client for each node in the graph, but for each round-trip, the static gates are applied multiple times. Any imperfections in the alignment, which may be smaller than for the reconfigurable controllers, will therefore be multiplied.

For the reconfigurable controllers, there is often a trade-off between the controller speed, as discussed in the last section, and the precision of its control. This is not a fundamental trade-off and more of an engineering challenge, though a trade-off that is often seen regardless. 

\textbf{Multi-photon events. } We assume all single photon emitters are true single photon sources. This is because probabilistic multi-photon events cause security problems. In \cite{garnier2024composably}, this security concern is patched in the context of WCP sources. With WCP-based clients, the multi-photon emissions do not only cause a security-induced protocol overhead, as discussed in Section \ref{subsec:security}, but also errors. At a BSM, the multi photon events can cause false heralding similar to dark counts. 

\textbf{Decoherence. } Decoherence of the server memory qubits is present for all architectures, but in varying degrees. For most architectures, the qubits in memory will decohere for as long as it takes for their neighbours in the graph state to be remotely prepared if $p\gg n_\mathrm{r}/c$, recall that $n_\mathrm{r}$ is the number of rows, and $c$ is the cutoff. If instead $p\ll n_\mathrm{r}/c$, the cutoff will often be triggered, and the average time spent in memory will approach the cutoff time itself. How much decoherence you `allow' to happen can be set directly by choosing $c$. This rate-fidelity trade-off occurs for all architectures except for the direct measurement client.

In the direct measurement approach the creation of the qubits happen locally, which is assumed to happen deterministically and take $t_\mathrm{init}$. Qubits will not idle in memory for multiple attempts, but instead are re-initiated after every attempt of sending the graph. Because of this, the qubits will undergo relatively little decoherence in this approach. While this sounds like a positive aspect compared to the other approaches, we note that one can achieve the same in the other architectures by choosing to set the cutoff to $c=1$, though this is often not an optimal choice given the detrimental effect this has on the rate.

\textbf{Compounding errors. } We briefly mentioned that in the fixed-gate clients, the client control error compounds, because the gates are traversed multiple times. In addition to this, the whole sequence of emission, applying gates, and transferring the state back onto the matter qubit, is also traversed multiple times for the multi-pass manipulation clients: twice for the $T$-gate client and three times for the $T/X$-gate client. Because of this, all the errors that occur along the way compound. This is also true for the WCP-based client, where multiple qubits need to be concatenated into one for each qubit in the graph, compounding the errors of all qubits $k$ in the gadget.

\textbf{Overview. } Now that we have discussed all the errors that occur for each architecture and their corresponding protocols, let us collect them into an overview. We collect all errors in Table \ref{tab:errors}.

\subsection{Hardware cost and complexity}\label{subsec:cost}

The advantage of a client architecture can be something different than how it performs or what security guarantees it provides. Specifically, the client device is supposed to be a simple, cheap device relative to the quantum server. In addition, some architectures put requirements on the server, and can thus not be used in combination with any setup. We discuss these aspects, as well as general design considerations for the client devices, below.

\textbf{Photon source considerations for emission-based clients.} For BSM-based remote state preparation, the interference visibility, and hence the fidelity of the prepared state, depends critically on the indistinguishability of the client and server photons. Achieving high spectral, temporal, and spatial overlap typically requires that both photons originate from similar physical systems, pushing the client towards hardware comparable to the server and thereby defeating the purpose of delegating a computation in the first place. WCP sources are substantially cheaper than true single-photon sources, and they have the additional upside that their properties can to a greater extent be shaped through external modulation to better match the server's emission profile.

For cavity reflection, the constraints shift from photon-photon indistinguishability to photon-matter coupling. The client's photon must be resonant with the server's atomic transition and have a bandwidth within the cavity linewidth, while maintaining high optical coherence.

\textbf{Server capability considerations.} The hardware requirements at the receiving node differ substantially between the two cavity-reflection schemes. CAPS~\cite{kikura2025passive} demands only that the server provide a qubit with one cavity-coupled and one cavity-uncoupled state, a structure available in essentially any matter qubit with a spin- or state-selective optical transition, including neutral atoms, trapped ions, NV and SiV centers, and charged quantum dots. SPRINT~\cite{rosenblum2017analysis} is more restrictive: it requires a three-level $\Lambda$-system in which the two legs couple to *different* propagation modes of the cavity, so that counter-propagating photons address different atomic transitions. In practice this has been realized using chiral light–matter coupling in whispering-gallery-mode resonators with TM modes, restricting SPRINT-compatible nodes to platforms where such directional coupling can be engineered, primarily atoms (or, in principle, other emitters) evanescently coupled to WGM microresonators or nanophotonic waveguides supporting spin-momentum locking. This narrows the set of viable hardware platforms compared to CAPS.

For rotation-based clients, the server also acts as a single-photon source (sending qubit state $|0\rangle$ for the $H/Z(\pi/4)$ and $T/X$ clients and $|+\rangle$ for the $Z$-rotation client and the $T$-gate client), in addition to the role it takes in other prepare-and-send-type protocols. Specifically, because the rotation-based clients do not have a memory, the server needs to be able to emit at a sufficiently high rate. For a rotation-based client that uses BSM-based teleportation to transfer the state back into the server memory, the server needs to emit a memory-entangled photon towards the BSM exactly one communication time ($t_\text{com}$) after emitting the photon towards the client, thus we need $t_\text{emit}^{(s)}\leq t_\text{com}$.

In BSM-based teleportation approaches, the server relies on generating matter–photon entanglement. By contrast, in direct measurement approach, the server must coherently prepare and transmit a photonic graph state to the client by mapping matter qubits onto photonic qubits (e.g., via SWAP operations). This requires the ability to emit sequences of photons with a stable phase relation and to preserve entanglement across multiple emissions, which can be a difficult engineering challenge.

\textbf{Fast control of state manipulation.} In prepare-and-send protocols, the client must prepare states $\ket{+_\theta}$ with $\theta \in \{k\pi/4\}^{7}_{k=0}$, selecting a fresh random angle for each attempt to maintain blindness. The $Z$-rotation client of Ref.~\cite{kashefi2024verification} has this same requirement. This angular control can be implemented using waveplates or liquid crystal retarders. This is often a trade-off in choosing the hardware: more precise waveplate setups often have lower switching speeds and vice versa. Lower switching speeds means that the attempts are slower, leading to a lower rate and more decoherence, while lower precision leads to a lower fidelity of the remotely prepared state. This trade-off is not fundamental: more precise and fast state manipulation can also be implemented with for example EOMs \cite{drmota2024verifiable}, though this increases the clients hardware cost and complexity.

The impact of such a trade-off is amplified in round-based verification protocols such as that of Leichtle et al.~\cite{leichtle2021verifying}, where the total number of rounds required to reach a target security level scales with the per-round error rate. Reduced state preparation fidelity therefore translates directly into longer end-to-end runtimes, potentially negating gains from faster modulation. Optimising this trade-off requires minimising total protocol runtime as a function of both switching speed and state preparation fidelity; this problem is addressed in~\cite{vandam2026optimizing}.

In the direct measurement approach, blindness is protected by the no-signalling principle, and the client thus does not have to change its measurement basis for each attempt. Instead, it directly measures each qubit in the basis (depending on $\phi$) that defines the calculation, thus it must change its measurement basis for each qubit in the graph. However, for this setup, all photons of a given graph need to reach the client, i.e., there is no re-try per qubit. We discuss this further in the next section. Therefore, the basis still needs to change for every photon that is sent and the same requirements apply. The requirements on the switching speed might even be increased since we do not need to classically communicate success/failure for each qubit, only for each graph, meaning that the qubits can be sent closer together in time, as long as the servers emission rate allows. 

The architectures that do not need fast variable rotators are the fixed-gate clients of Wu et al. and Li et al., which use one or two fixed gate elements through which photons pass a variable number of times, with a fast optical switch required to control the loop count. In addition to the gate loop, the client would need a delay line it can loop through such that the total time spent at the client is independent of the number of gates that are applied, such as to not leak information to the server. Since the rotation angle is fixed, this eliminates the speed-precision trade-off for angular control, albeit with some requirement on the speed of the switch. This might provide an advantage particularly when fabricated onto a chip with integrated photonics.

The requirement on how fast the client should be able to change its basis depends on the time-per-attempt, which is discussed in Section \ref{subsec:rate}. For example, for a measurement-RSP client, this comes down to the emission time of the server plus the round-trip communication time between the client and server, combined usually sub-ms (for 50 km, the one-way communication time is 0.25 ms). The server can slow down its emission rate to accommodate for slower clients, though this of course comes at the cost of reduced rate.

We note that if the static gates with controlled loops are easier to implement in a given situation, this approach can also be used for the other client types. This is because we do not need to be able to rotate to arbitrary states, just to the 8 different states on the equatorial plane of the Bloch sphere ($|\pm_\theta\rangle,\,\theta\in\{\frac{k\pi}{4}\}_{0= k}^7$), and sometimes $|0/1\rangle$, which can be implemented with looping through static gates, such as $H$ and $T$. The other way around is also true: the fixed-gate client architectures can be implemented with variable controllers too, though this would go against the goal of the protocols to give as little capabilities to the client as possible.

\textbf{Client architecture complexity comparison.} It is important for the client device to be as simple as possible, such that the number of users that can interact with the network can be increased. The most simple, low-cost implementation would be a rotation-based client. This client type only has to manipulate a photonic quantum state, something that emission-based and measurement-based clients also have to do in addition to being able to measure or emit photons in specific bases. The fixed-gate clients are possibly simpler than the $Z$-rotation client, depending on the exact implementation and the $T$-gate client has the simplest architecture.

Beyond architectural complexity, the cost of the components themselves influences the practical accessibility of each client type. As discussed before, all architectures require photonic state manipulation, which can be done by inexpensive elements such as waveplates, though limiting the rate of execution, or more expensive elements such as EOMs. Then, adding onto this, at the low end, weak coherent pulse sources (attenuated lasers) are inexpensive, commercially available, and require minimal infrastructure. Single-photon detectors occupy a middle ground: while high-performance detectors such as superconducting nanowire single-photon detectors (SNSPDs) can be costly and require cryogenic operation, the detection task in BQC does not demand extreme timing resolution or efficiency (though lower efficiency is more detrimental in the direct measurement case, as we will see in the next section), making more affordable options like avalanche photodiodes (APDs) viable for many implementations. Single-photon sources, by contrast, sit at the expensive end of the spectrum. The cost here is driven not only by the source itself but by the matching requirements discussed earlier in this section.

\textbf{Summary.} Rotation-based clients have the simplest architecture, and each of them can be implemented with similar architectures: either with variable photonic state controllers, or fixed controllers and looping. The complexity of the client increases when introducing a WCP, measurement device or single photon emitter. The single photon emitter, crucially, has to emit photons that interact well with the photons coming from the server, introducing either visibility errors and/or technical overhead due to mitigation techniques. The cavity-reflection can only be used when the server actually has a cavity, with exact requirements depending on the technique. 

\section{Discussion}
We have introduced a client architecture taxonomy and analyzed each option across security levels and overhead, losses, rate and multiplexing, error and failure modes, and design considerations. Overviews are provided in Tables \ref{tab:security}, \ref{tab:rate}, and \ref{tab:errors}. Here, we draw out conclusions that emerge from viewing these dimensions together.

Noise robustness of the security proof acts as a practical requirement for near-term deployment, and immediately narrows the field when strong security guarantees are required. The direct measurement and fixed-gate clients both lack noise-robust verification schemes; their hardware simplicity therefore does not currently translate into deployment readiness. In addition, the rate scaling of the direct measurement approach is detrimental for large graphs. If lowering security guarantees is acceptable, for example by not requiring information-theoretic security, classical solutions exist \cite{mahadev2020classical, gheorghiu2019computationally, mahadev2018classical}.

Among the architectures with noise-robust verification, no single option dominates across all dimensions. The single-photon emission client that uses BSMs has a design tension: it requires high HOM visibility, meaning photons from the client and server should be indistinguishable, pushing the client toward server-like hardware and undermining the delegation premise. WCP sources partially escape the matching requirement, but at the cost of multi-photon events and a gadget construction to handle them from a security perspective, which carries its own costs in compounding errors and time-overhead. Rotation-based clients have a very simple architecture, but the round-trips are detrimental for rate, increasing losses and time per attempt while limiting multiplexing options and compounding errors. This disadvantage grows with client-server distance.

Two architectures stand out as particularly well-suited for near-term deployment: the single-photon emitter with cavity-reflection-based interaction at the server, and the measurement-based RSP client. Both have relatively few failure modes and avoid the photon-photon indistinguishability requirements of the BSM approach. The choice between them largely reduces to which hardware is more accessible in a given setting: cavity reflection requires a high-quality cavity and photon-matter interface at the server, and a client emitter whose photons are spectrally matched to the server's transition; measurement-based RSP shifts complexity to the client's measurement apparatus while removing the need for a client-side emitter altogether.

Ultimately, the right choice of client architecture depends on which constraints dominate. If strong security guarantees are essential, noise-robust verification is required and the field narrows accordingly; otherwise, other clients become an option. Among the noise-robust options, the trade-off depends on the available hardware — whether a high-quality photon-matter interface, a spectrally compatible single-photon emitter, or a capable measurement setup is most readily available — as well as on the client-server distance. If minimizing client-side hardware cost is the overriding priority, rotation-based clients remain attractive despite their performance disadvantages, particularly at short distances. As hardware matures and security frameworks develop, the balance between these considerations will shift, and the conclusions drawn here should be revisited accordingly.

\section*{Acknowledgements}
The authors would like to thank Sergio Loarte and Conor Bradley for useful discussions on reflection-based schemes, Harrold Ollivier and Maxime Garnier for useful discussions on the security of VBQC protocols, and the labs of Tracy Northup, Ben Lanyon, Tim Taminiau and Ronald Hanson for many useful discussions on practical implementability and server-side considerations. The authors additionally thank Soubhadra Maiti for useful comments on the manuscript.\\
This project has received funding from the European Union's Horizon Europe research and innovation programme under grant agreement No.\ 101102140.

\bibliographystyle{unsrt}
\bibliography{bib}

\appendix
\include{Appendix}
\end{document}

%% file: client_taxonomy.tex
\begin{tikzpicture}

\tikzset{
  foldtri/.style={
    regular polygon, regular polygon sides=3,
    fill=rorange!60, draw=rorange!80!black,
    inner sep=2.2pt, shape border rotate=-90,
    font=\tiny
  },
  foldtri em/.style={
    regular polygon, regular polygon sides=3,
    fill=emgreen!60, draw=emgreen!80!black,
    inner sep=2.2pt, shape border rotate=-90,
    font=\tiny
  },
}

\node[mainbox=emgreen, anchor=west] (emission) at (0, 0) {Emission-based};

\node[subbox=emgreen, anchor=west] (wcp)
    at (\hind, -0.9) {WCP (BSM)};

\node[subbox=emgreen, anchor=west] (sp)
    at (\hind, -1.9) {Single-photon};
\node[foldtri em] (spfold) at ($(sp.east)+(0.4,0)$) {};

\def\etrx{0.45*\hind}
\draw[garr] (emission.south -| \etrx, 0)  |- (wcp.west);
\draw[garr] (emission.south -| \etrx, 0)  |- (sp.west);

\node[mainbox=meblue, anchor=west]
    (measurement) at (0, -3.4) {Measurement-based};

\node[subbox=meblue, anchor=west] (measrsp)
    at (\hind, -4.3) {Measurement RSP};

\node[subbox=meblue, anchor=west] (direct)
    at (\hind, -5.3) {Direct measurement};

\draw[barr] (measurement.south -| \etrx, 0)  |- (measrsp.west);
\draw[barr] (measurement.south -| \etrx, 0)  |- (direct.west);

\node[dashbox, inner sep=3pt,
      fit=(sp)(spfold)(measrsp)] (psbox) {};
\node[font=\footnotesize\itshape, text=gray!75!black, anchor=west, align=left]
    at ($(psbox.east)+(0.15,0)$) {Prepare-and-send};

\node[dashbox, inner sep=3pt, fit=(direct)] (rmbox) {};
\node[font=\footnotesize\itshape, text=gray!75!black, anchor=west, align=left]
    at ($(rmbox.east)+(0.15,0)$) {Receive-and-measure};

\node[dashbox, inner sep=3pt, fit=(wcp)] (wcpbox) {};
\node[font=\footnotesize\itshape, text=gray!75!black, anchor=west, align=left]
    at ($(wcpbox.east)+(0.15,0)$) {Reduces to\\prepare-and-send};

\node[mainbox=rorange, anchor=west]
    (rotation) at (\Rx, 0) {Rotation-based};

\node[subbox=rorange, anchor=west] (singlepass)
    at ({\Rx+\hind}, -0.9) {Single-pass manipulation};

\node[subbox=rorange, anchor=west] (zrot)
    at ({\Rx+2*\hind}, -1.9) {Z-rotations};
\node[foldtri] (zfold) at ($(zrot.east)+(0.4,0)$) {};

\node[subbox=rorange, anchor=west] (htgate)
    at ({\Rx+2*\hind}, -2.9) {$H/Z(\pi/4)$-gate};
\node[foldtri] (htfold) at ($(htgate.east)+(0.4,0)$) {};

\node[subbox=rorange, anchor=west] (multipass)
    at ({\Rx+\hind}, -4.1) {Multi-pass manipulation};

\node[subbox=rorange, anchor=west] (tgate)
    at ({\Rx+2*\hind}, -5.0) {$T$-gate};
\node[foldtri] (tfold) at ($(tgate.east)+(0.4,0)$) {};

\node[subbox=rorange, anchor=west] (txgate)
    at ({\Rx+2*\hind}, -6.0) {$T/X$-gate};
\node[foldtri] (txfold) at ($(txgate.east)+(0.4,0)$) {};

\def\rtrx{\Rx+0.45*\hind}
\draw[oarr] (rotation.south -| \rtrx, 0)  |- (singlepass.west);
\draw[oarr] (rotation.south -| \rtrx, 0)  |- (multipass.west);

\def\sptrxR{\Rx+\hind+0.45*\hind}
\draw[oarr] (singlepass.south -| \sptrxR, 0)  |- (zrot.west);
\draw[oarr] (singlepass.south -| \sptrxR, 0)  |- (htgate.west);

\draw[oarr] (multipass.south -| \sptrxR, 0)  |- (tgate.west);
\draw[oarr] (multipass.south -| \sptrxR, 0)  |- (txgate.west);

\node[dashbox, inner sep=3pt,
      fit=(singlepass)(zrot)(zfold)(htgate)(htfold)] (spbox) {};
\node[font=\footnotesize\itshape, text=gray!70!black,
      anchor=west, align=left]
    at ($(spbox.east)+(0.15,0)$) {Reduces to\\prepare-and-send};

\node[draw=rorange!50!gray, dotted, line width=1.0pt,
      rounded corners=5pt, inner sep=7pt,
      fit=(htgate)(htfold)(tgate)(tfold)(txgate)(txfold)] (fixedbox) {};
\node[font=\footnotesize\itshape, text=rorange!50!gray!90!black,
      anchor=west, align=left]
    at ($(fixedbox.east)+(0.15,0)$) {Fixed-gate};

\colorlet{leggray}{gray!55}

\coordinate (legorigin) at (0, -7.0);

\fill[gray!5, draw=gray!60, rounded corners=3pt, line width=0.6pt]
    ($(legorigin)+(-0.4, 0.45)$) rectangle ($(legorigin)+(5.2, -3.3)$);
\node[font=\footnotesize\bfseries, text=gray!70!black, anchor=south west]
    at ($(legorigin)+(-0.4, 0.45)$) {Legend};

\node[foldtri neutral] (legtri) at (legorigin) {};
\node[anchor=west, font=\small] (leglabel)
    at ($(legtri.east)+(0.15,0)$) {$=$ photon--matter interface:};

\node[subbox=leggray, anchor=west, font=\small] (legcav)
    at ($(legorigin)+(0.5, -0.8)$) {Cavity reflection};
\node[subbox=leggray, font=\scriptsize, inner ysep=1.5pt, anchor=west] (legsprint)
    at ($(legorigin)+(1.7, -1.5)$) {SPRINT};
\node[subbox=leggray, font=\scriptsize, inner ysep=1.5pt, anchor=west] (legcaps)
    at ($(legorigin)+(1.7, -2.1)$) {CAPS};
\node[subbox=leggray, anchor=west, font=\small] (legbsm)
    at ($(legorigin)+(0.5, -2.8)$) {BSM RSP};

\coordinate (legtrunk) at ($(legorigin)+(0.2, 0)$);
\draw[arr, color=leggray!80!black, sharp corners]
    (legtrunk |- legtri.south) |- (legcav.west);
\draw[arr, color=leggray!80!black, sharp corners]
    (legtrunk |- legtri.south) |- (legbsm.west);

\coordinate (legcavtrunk) at ($(legorigin)+(1.45, 0)$);
\draw[arr, color=leggray!80!black, sharp corners]
    (legcavtrunk |- legcav.south) |- (legsprint.west);
\draw[arr, color=leggray!80!black, sharp corners]
    (legcavtrunk |- legcav.south) |- (legcaps.west);

\end{tikzpicture}

%% file: Appendix.tex
\onecolumngrid
\section{Expected delivery time of graph states with cutoffs}\label{app:rate_deriv}
In this appendix we study the expected time that a client needs to successfully create a graph state at a server. We assume that the client and server follow an attempt-based protocol where in each attempt the client attempts to prepare a single qubit at the server. Due to photon loss this succeeds with probability $p<1$. The order in which the qubits are prepared follows a predetermined flow. The \emph{measure-as-you-go} paradigm is used in which the server measures a qubit as soon as all its neighbours in the graph have been generated and the corresponding \textsc{cz}-gates have been applied. Crucially, since each qubit only has a limited coherence time there is a cutoff on the number of attempts it may be stored in memory. If the qubit cannot be measured before the cutoff, then it is discarded and the entire graph has to be reconstructed from scratch. We refer to this as a \emph{qubitwise cutoff} because the cutoff condition constrains the number of attempts any single qubit may be stored in memory. The cutoffs give the graph state generation process the structure of a \textit{renewal process} \cite[Chapter 3]{ross1995stochastic}. We use the renewal process structure and probability generating functions to compute the expected number of attempts the client needs to successfully prepare the whole graph state at the server within the constraints posed by the qubitwise cutoffs. For a linear graph this yields an analytic expression for the expected number of attempts, while for a brickwork graph we obtain upper- and lowerbounds in closed-form. The problem analysed in this appendix is related to that of the window problem studied in Ref. \cite{davies2023tools}.

\subsection{Renewal process formulation}
A renewal process is a counting process in which the interarrival times of events are independent and identically distributed with an arbitrary distribution \cite[Chapter 3]{ross1995stochastic}. For the graph state generation problem the events are the times at which graph states are successfully delivered, or at which a cutoff condition is violated. In both cases a new graph state will have to be created from scratch again. The process of generating the graph thus proceeds in cycles, where each cycle starts with no qubits in the graph, and ends when the graph is created or has to be reset. This is the renewal structure. We additionally split each cycle into an \textit{idle period }$I$ and \textit{busy period} $B$, where the idle period corresponds to the time until the first qubit is generated, and the busy period corresponds to the time when the server has to keep qubits in memory. More precisely, the busy period starts after the first qubit has been generated and it ends either when a qubit violates the cutoff condition, or when the whole graph has been generated. When a cutoff is violated we say the busy period has failed, while if whole graph is generated we say that the busy period has succeeded.

The number $K$ of single-qubit attempts until the first successful graph state generation is then distributed according to
\begin{equation}\label{eq: definition of K}
    K = \sum_{i=1}^N X_i + Y,
\end{equation}
where $X_i$ is the duration of the $i$-th failed cycle, $N$ is the total number of failed cycles, and $Y$ is the duration of the first successful cycle. In each cycle, the duration of the idle period is an independent geometric random variable $I\sim \Geo(p)$. The duration $B$ of the busy period depends on the outcome $O\in \{s,f\}$, where $O=s$ for a successful busy period and $O=f$ for a failed busy period. For the failed cycles we thus have $X_i = I_i + (B_i|O_i=f)$ and for the successful cycle we have $Y = I + (B|O=s)$. The following lemma gives the expectation value of $K$.
\begin{lemma}
    The expectation value of $K$ is given by
    \begin{equation}\label{eq: expectation K abstract}
    \EE(K) = \frac{1/p+\EE(B)}{p_{B}},
\end{equation}
where $p_{B}:=\PP(O=s)$ is the probability that the graph state is successfully constructed during the busy period, and $\EE(B)$ is the expected duration of the busy period, irrespective of whether it succeeded or failed. 
\end{lemma}
\begin{proof}
    Taking the expectation value in \cref{eq: definition of K}, it follows that
    \begin{equation}
        \EE(K) = \EE\Big(\sum_{i=1}^N X_i \Big) + \EE(Y).
    \end{equation}
    By Wald's equation (Theorem 3.3.2 in Ref. \cite{ross1995stochastic}), it follows that
    \begin{equation}\label{eq: wald}
        \EE(K) = \EE(N)\EE(X) + \EE(Y).
    \end{equation}
    With $p_B$ the success probability of the busy period, the expected number of failed busy periods until the first success is
    \begin{equation}
        \EE(N) = \frac{1}{p_B}-1,
    \end{equation}
    where $p_B= \PP(O=s)$.
    Using that $\EE(I)=1/p$, the duration of a failed cycle is
    \begin{equation}
        \EE(X) = \frac{1}{p} + \EE(B|O=f),
    \end{equation}
     and the duration of a successful cycle is
     \begin{equation}
         \EE(Y) = \frac{1}{p} + \EE(B|O=s).
     \end{equation}
     Hence, substituting into \cref{eq: wald} and using that $\EE(B) = (1-p_B)\EE(B|O=f)+p_B\EE(B|O=s)$,
     \begin{align}
         \EE(K) = \left(\frac{1}{p_B}-1\right)\left(\frac{1}{p}+\EE(B|O=f)\right)+\left(\frac{1}{p}+\EE(B|O=s)\right)
         =\frac{1/p+\EE(B)}{p_B},
     \end{align}
     which proves the result.
\end{proof}

\subsection{Linear graph}
Consider a linear graph of $n$ qubits, indexed $i=1,\dots, n$. The flow is to generate the qubits in order of increasing index. Let $c$ denote the qubitwise cutoff on the maximum number of attempts any qubit may be stored in memory. If the next qubit in the flow has not been generated within $c$ attempts, the graph is discarded and must be constructed from scratch again. We apply \cref{eq: expectation K abstract} to compute the expected number of attempts until the first successful graph is created. This requires the computation of the success probability of the busy period and its expected duration.

During the busy period, a qubit is held in memory while attempts at the next qubit are being made.
The attempts to create qubit $j+1$ when the flow is currently at qubit $j\in \{1,\dots,n-1\}$ are described by a pair of random variables $(K_j,O_j)$, where $K_j\in \{1,\dots, c\}$ gives the number of attempts to create qubit $j+1$ with qubit $j$ in memory and $O_j\in \{s,f\}$ indicates the outcome success $(s)$ if the qubit $j+1$ is generated within $c$ attempts, or otherwise failure $(f)$. These random variables have the joint distribution
\begin{equation}\label{eq: distribution K and O linear graph}
    \PP(k, s) = p(1-p)^{k-1} \quad \mathrm{for}\quad k\in \{1,\dots, c\}\qquad \mathrm{and} \qquad \PP(c,f)=(1-p)^c,
\end{equation}
so that the outcome is a success if and only there is a successful generation within $c$ attempts. 
The total busy period is successful if and only if $O_j=s$ for each $j=1,\dots,n-1$. Hence,
\begin{equation}\label{eq: success linear graph}
    p_B = \PP(O_1=s)\cdots \PP(O_{n-1}=s) = (1-(1-p)^c)^{n-1}.
\end{equation}

To compute the expected duration $\EE(B)$ of the busy period we will derive its \textit{probability generating function}
\begin{equation}
    G_B(x):= \EE\big(x^B\big).
\end{equation}

\begin{proposition}\label{prop: generating fucntion linear graph}
    The probability generating function for the duration $B$ of the busy period for a linear graph with $n$ qubits and qubitwise cutoff $c$ is given by
    \begin{equation}\label{eq: generating function linear graph}
        G_B(x) = S_1(x)\cdots S_{n-1}(x)+
        \sum_{j^*=1}^{n-1}S_1(x)\cdots S_{j^*-1}(x)F_{j^*}(x),
    \end{equation}    
    where 
    \begin{equation}\label{eq: definition defective generating function linear graph}
    S_j(x):= \EE\big(x^{K_j}\indicator(O_j=s)\big) \qquad \mathrm{and}\qquad F_j(x):= \EE\big(x^{K_j}\indicator(O_j=f)\big).
    \end{equation}
    are the defective generating functions for $K_j$ for $O_j=s$ or $O_j=f$, respectively.
\end{proposition}
\begin{proof}
    By definition, the probability generating function for the duration $B$ of the duration of the busy period for a linear graph with $n$ qubits and qubitwise cutoff $c$ is $G_B(x)=\EE(x^{B})$. Let $(K_j,O_j)\overset{\mathrm{iid}}{\sim}\TruncGeo(p,c)$ denote the numbers of attempts spent at qubits $j=1,\dots,n-1$ in the busy period, then
    \begin{equation}
        B \sim \Big(\sum_{j=1}^{n-1}K_j\Big) \indicator(O_1=s)\cdots \indicator(O_{n-1}=s) + \sum_{j^*=1}^{n-1}\Big(\sum_{j=1}^{j^*}K_j\Big)\indicator(O_1=s)\cdots \indicator(O_{j^*-1}=s)\indicator(O_{j^*}=f) 
    \end{equation}
    since the busy period ends when there is success at each qubit, or at the first failure at $j^*$. Since all $(K_j,O_j)$ are independent it follows that
    \begin{equation}
        \EE\big(x^B\big) = \prod_{j=1}^{n-1}\EE\big(x^{K_j}\indicator(O_j=s)\big) + \sum_{j^*=1}^{n-1}\Big(\prod_{j=1}^{j^*-1}\EE\big(x^{K_j}\indicator(O_j=s)\big)\Big)\EE\big(x^{K_{j^*}}\indicator(O_{j^*}=f)\big).
    \end{equation}
    This yields \cref{eq: generating function linear graph} after substituting the definitions of the defective generating functions given in \cref{eq: definition defective generating function linear graph}.
\end{proof}
Using the distribution of $(K_j,O_j)$ in \cref{eq: distribution K and O linear graph} the defective generation functions in \cref{eq: definition defective generating function linear graph} for successful and failed generation at qubit $j$ are given in closed form by
\begin{equation}
    S_j(x) = \EE\big(x^{K_j}\indicator(O_j=s)\big) = \sum_{k=1}^c p(1-p)^{k-1} x^k = \frac{px (1-(1-p)^cx^c)}{1-(1-p)x}
\end{equation}
and
\begin{equation}
     F_j(x) = \EE\big(x^{K_j}\indicator(O_j=f)\big) = (1-p)^{c} x^c.
\end{equation}
The expected duration of the busy period can be obtained from the probability generating function as
\begin{equation}\label{eq: expected busy period linear graph}
    \EE(B) = \dv{x}G_B(1) = \frac{1}{p}\sum_{j^*=1}^{n-1}(1-(1-p)^c)^{j^*}.
\end{equation}
We conclude using \cref{eq: expectation K abstract}, \cref{eq: success linear graph}, and \cref{eq: expected busy period linear graph} that the expected number of attempts to generate a linear graph state with $n$ qubits subject to qubitwise cutoff $c$ is
\begin{equation}\label{eq: expected number of attempts linear graph}
    \EE(K) = \frac{\frac{1}{p} +\frac{1}{p}\sum_{j^*=1}^{n-1}(1-(1-p)^c)^{j^*}}{(1-(1-p)^c)^{n-1}}.
\end{equation}

In the regime of $cp\ll 1$ the expected duration of the busy period is $\EE(B) \approx  c,$
and the expected number of attempts to generate the whole graph state is approximately
\begin{equation}
    \EE(K) \approx \frac{(1/p)+c}{(cp)^{n-1}}\approx \frac{1}{c^{n-1}p^{n}}.
\end{equation}
For small $p$, the cutoff thus gives a multiplicative boost to the success probability of generating the next qubit in the graph during the busy period. 
If the cutoff goes to infinity the expected duration of the busy period is $\lim_{c\to \infty}\EE(B)=\frac{n-1}{p}$. Hence, in the limit $c\to \infty$, \cref{eq: expected number of attempts linear graph} correctly recovers $\lim_{c\to\infty}\EE(K) = n/p$ which corresponds to the expectation of the sum of $n$ geometric random variables with parameter $p$, each of which gives the number of attempts to generate one of the qubits.

\subsection{Brickwork graph}

A graph state of particular interest in blind quantum computing is a brickwork graph. We consider a brickwork graph with $n_\mathrm{r}$ rows and $n_\mathrm{c}$ columns for a total of $n=n_\mathrm{r}n_\mathrm{c}$ qubits. In the measure-as-you-go flow the qubits are generated from top to bottom in each column, and the columns are generated from left to right. In other words, labelling the qubits by tuples $(i,j)$ with $i=1,\dots, n_\mathrm{r}$ and $j=1,\dots, n_\mathrm{c}$ they are generated in a lexicographic order, where column number takes precedence over row number. We say that $(i,j)\leq (i',j')$ if $(i,j)$ precedes $(i',j')$ in the flow.
We let $K_{i,j}$ denote the number of attempts while at qubit $(i,j)$ in the flow. For qubits in the last column $j=n_\mathrm{c}$, the \textsc{cz}-gate and measurement occur directly after the qubit is generated, while for $j<n_\mathrm{c}$, qubit $(i,j)$ can only be measured after $(i,j+1)$ is generated and the \textsc{cz}-gate is performed. The qubitwise cutoff $c$ imposes the constraint that the graph is successfully generated if and only if
\begin{equation}\label{eq: brickwork B constraint}
    \sum_{i'= i+1}^{n_\mathrm{r}}K_{i',j-1} + \sum_{i'=1}^{i} K_{i', j} \leq c \quad \forall (i,j)\leq (n_\mathrm{r}, n_\mathrm{c}),
\end{equation}
where $K_{n_\mathrm{r},n_\mathrm{c}}=0$ because at the last qubit the graph is finished.
If the cutoff condition is violated at some qubit $(i^*,j^*)<(n_\mathrm{r}, n_\mathrm{c})$, then the process of building the graph has to start all over again. 

As for the linear graph, we model this as a renewal process with idle period $I\sim \Geo(p)$ corresponding to no qubits in the graph, and busy period $B$ when there are qubits in the graph. The expected number of attempts to construct the full graph can be found using \cref{eq: expectation K abstract} and requires the computation of the success probability $p_B$ of the busy period and its expected duration $\EE(B)$. To compute these quantities, we define a \emph{path in the busy period}, or simply, a \emph{path}, to be a vector $\kvec=(k_{i,j})_{(i,j)<(n_\mathrm{r},n_\mathrm{c})} \in \NN^{n-1}$, where $k_{i,j}$ is the number of attempts at qubit $(i,j)$ until success. We think of each $k_{i,j}$ as the outcome of a geometric distribution, without regards for any cutoff conditions. More precisely, we define the probability that a path $\kvec\in \NN^{n-1}$ occurs by 
\begin{equation}\label{eq: success of path}
     \PP(\kvec):=p^{n-1}(1-p)^{\big(\sum_{(i,j)<(n_\mathrm{r},n_\mathrm{c})}k_{i,j}\big)-(n-1)},
\end{equation}
where the total number of attempts in the path is $\sum_{(i,j)<(n_\mathrm{r},n_\mathrm{c})}k_{i,j}$, of which $n-1$ are successful so as to reach the final qubit (the first qubit was already generated to start the busy period). 

We let $\mathcal{P}_{(i^*,j^*)}(B)$ denote the set of all paths that terminate at qubit $(i^*,j^*)$. Paths that terminate at $(i^*,j^*)=(n_\mathrm{r},n_\mathrm{c})$ successfully create the whole graph, while paths that terminate at $(i^*,j^*)<(n_\mathrm{r},n_\mathrm{c})$ violate the cutoff condition for the first time at $(i^*,j^*)$. More precisely, $\kvec \in \mathcal{P}_{(i^*, j^*)}(B)$ for $(i^*,j^*)<(n_\mathrm{r}, n_\mathrm{c})$ if and only if
\begin{equation}\label{eq: condition fail B}
    \sum_{i'=i+1}^{n_\mathrm{r}}k_{i',j-1} + \sum_{i'=1}^{i}k_{i',j} \leq c \quad \forall (i,j)<(i^*,j^*) \quad \mathrm{and}\quad \sum_{i'=i^*+1}^{n_\mathrm{r}}k_{i',j^*-1} + \sum_{i'=1}^{i^*}k_{i',j^*} > c,
\end{equation}
so that the cutoff condition is satisfied at all qubits up to and including $(i^*-1,j^*)$, but is violated at $(i^*,j^*)$. The set of successful paths is $\mathcal{P}_{(n_\mathrm{r}, n_\mathrm{c})}(B)$. We have $\kvec \in \mathcal{P}_{(n_\mathrm{r}, n_\mathrm{c})}(B)$ if and only if
\begin{equation}\label{eq: condition succ B}
    \sum_{i'=i+1}^{n_\mathrm{r}}k_{i',j-1} + \sum_{i'=1}^{i}k_{i',j} \leq c \quad \forall (i,j)<(n_\mathrm{r},n_\mathrm{c}),
\end{equation}
so that the cutoff condition is never violated.

The success probability of the busy period is obtained by summing over all paths that successfully create the whole graph, and weighing them by their probabilities,
\begin{equation}
    p_{B} = \sum_{\kvec\in \mathcal{P}_{(n_\mathrm{r},n_\mathrm{c})}(B)} \PP(\kvec).
\end{equation}
The expected duration $\EE(B)$ of the busy period requires also taking into account paths where the cutoff condition is violated at some point. It can be expressed as
\begin{equation}
    \EE(B) = \sum_{(i^*,j^*)\leq  (n_\mathrm{r},n_\mathrm{c})} \sum_{\kvec \in \mathcal{P}_{(i^*,j^*)}(B)}\PP(\kvec) B(\kvec;c),
\end{equation}
where 
\begin{equation}\label{eq: duration B}
    B(\kvec;c) = 
    \begin{cases}
        \sum_{(i,j)\leq (i^*,j^*-1)} k_{i,j} + c & \mathrm{if\,}(i^*,j^*)<(n_\mathrm{r},n_\mathrm{c}) \\
        \sum_{(i,j)\leq (n_\mathrm{r},n_\mathrm{c})} k_{i,j} & \mathrm{if\,}(i^*,j^*)=(n_\mathrm{r},n_\mathrm{c})
    \end{cases}
\end{equation}
is the duration of the busy period $B$ for a path $\kvec\in \mathcal{P}_{(i^*,j^*)}(B)$. The qubit $(i^*,j^*)$ at which the path terminates can be determined uniquely from $\kvec$ and $c$ so that \cref{eq: duration B} is well-defined. 

Evaluating $p_{B}$ and $\EE(B)$ directly is hard because there are many paths to sum over. A priori an infinite amount of them in the above formulation. In principle, the sums can be made finite by summing the tail probabilities of the failed path. But even then, there are a lot of paths to sum over. Enumerating all paths is also complicated because the cutoff conditions are not independent, as they were in the linear graph. What we will do instead is to derive lowerbounds and upperbounds to $p_{B}$ and $\EE(B)$ which lead to bounds $\EE(K)$ through \cref{eq: expectation K abstract}.

\subsection{Bounds for brickwork graph}
The idea for obtaining lowerbounds to $p_B$ and $\EE(B)$ is to consider more stringent cutoff conditions, and to obtain upperbounds by relaxing the cutoff conditions, in such a way that there are less dependencies between cutoff conditions on different qubits. The reduction in dependencies is achieved by using \emph{columnwise cutoff} in which the cutoffs are on the number of attempts per column. This actually gives the problem the same structure as the linear graph, but then with columns instead of single qubits to be generated.

Define another renewal process with busy period $B'$ where the graph is reset if a column takes longer than $c$ attempts to be generated. We again consider paths $\mathcal{P}_{(i^*,j^*)}(B')$ that terminate at $(i^*,j^*)$ but now with respect to this columnwise cutoff condition. We have that $\kvec \in \mathcal{P}_{(n_\mathrm{r},n_\mathrm{c})}(B')$ if and only if
\begin{equation}\label{eq: condition succ B'}
    \sum_{i'=1}^{n_\mathrm{r}}k_{i',j}\leq c \quad\forall j\leq n_\mathrm{c},
\end{equation}
where we define $k_{n_\mathrm{r}, n_\mathrm{c}}=0$.
These are the paths for which the graph state is successfully created.
For $(i^*,j^*)<(n_\mathrm{r},n_\mathrm{c})$ the generation process fails, and we have that $\kvec \in \mathcal{P}_{(i^*,j^*)}(B')$ if and only if
\begin{equation}\label{eq: condition fail B'}
    \sum_{i'=1}^{n_\mathrm{r}}k_{i',j}\leq c \quad\forall j<j^*,\quad \sum_{i'=1}^{i^*-1}k_{i',j^*} \leq c,\quad \mathrm{and}\quad \sum_{i'=1}^{i^*}k_{i',j^*}>c.
\end{equation}

The success probability of busy period $B'$ with columnwise cutoffs is given by
\begin{equation}
    p_{B'} = \sum_{\kvec\in \mathcal{P}_{(n_\mathrm{r},n_\mathrm{c})}(B')} \PP(\kvec),
\end{equation}
where $\PP(\kvec)$ is given as before by \cref{eq: success of path}. The dependence on the cutoff condition is only in which path are summed over, not in the probability of a path.

We use the columnwise cutoff condition in $B'$ with cutoffs $c/2$ and $c$ to bound the success probability $p_{B}$ for qubitwise cutoffs $c$. To distinguish the different values of the cutoffs, we will denote the busy period of columnwise cutoffs $c$ as $B'(c)$ and the busy period of qubitwise cutoffs $c$ as $B(c)$. 
The following lemma shows that any path that is successful in busy period $B'(c/2)$ is also successful in busy period $B(c)$, and that any path that is successful in busy period $B(c)$ is also successful in busy period $B'(c)$.
\begin{lemma}\label{lemma: succesful path inclusion lemma}
The paths for successful busy periods of $B'(c/2)$, $B(c)$, and $B'(c)$ satisfy
\begin{equation}
    \mathcal{P}_{(n_\mathrm{r},n_\mathrm{c})}(B'(c/2)) \subseteq \mathcal{P}_{(n_\mathrm{r},n_\mathrm{c})}(B(c)) \subseteq \mathcal{P}_{(n_\mathrm{r},n_\mathrm{c})}(B'(c)) .
\end{equation}
\end{lemma}
\begin{proof}
We start by showing $\mathcal{P}_{(n_\mathrm{r},n_\mathrm{c})}(B'(c/2)) \subseteq \mathcal{P}_{(n_\mathrm{r},n_\mathrm{c})}(B(c))$. Let $\kvec \in \mathcal{P}_{(n_\mathrm{r},n_\mathrm{c})}(B'(c/2))$. Then,
\begin{equation}
    \sum_{i=1}^{n_\mathrm{r}}k_{i,j} \leq c/2 \quad \forall j\leq n_\mathrm{c}.
\end{equation}
Hence, for any $(i,j)\leq (n_\mathrm{r},n_\mathrm{c})$, it follows by the above that
\begin{equation}
    \sum_{i'=i+1}^{n_\mathrm{r}} k_{i',j-1} + \sum_{i'=1}^{i}k_{i',j}
    \leq c/2 + c/2
    \leq c.
\end{equation}
This means that $\kvec \in \mathcal{P}_{(n_\mathrm{r},n_\mathrm{c})}(B(c))$. We conclude that $\mathcal{P}_{(n_\mathrm{r},n_\mathrm{c})}(B'(c/2)) \subseteq \mathcal{P}_{(n_\mathrm{r},n_\mathrm{c})}(B(c))$.

Now for the second inclusion, let $\kvec \in \mathcal{P}_{(n_\mathrm{r},n_\mathrm{c})}(B(c))$. Then,
\begin{equation}
    \sum_{i'=i+1}^{n_\mathrm{r}} k_{i',j-1} + \sum_{i'=1}^{i}k_{i',j}\leq c\quad \forall (i,j)<(n_\mathrm{r},n_\mathrm{c}).
\end{equation}
With $i=n_\mathrm{r}$, the above yields
\begin{equation}
    \sum_{i=1}^{n_\mathrm{r}} k_{i,j}\leq c, 
\end{equation}
for any $j\leq n_\mathrm{c}$,
so that $\kvec \in \mathcal{P}_{(n_\mathrm{r},n_\mathrm{c})}(B'(c))$. This proves the second inclusion.
\end{proof}
\begin{corollary}[Success probability bounds]\label{cor: success probability bounds}
    The success probabilities of the busy periods $B'(c/2)$, $B(c)$, and $B'(c)$ satisfy
    \begin{equation}
        p_{B'(c/2)}\leq p_{B(c)} \leq p_{B'(c)}.
    \end{equation}
\end{corollary}
\begin{proof}
    The success probability in each case is of the form $
        p_B=\sum_{\kvec \in \mathcal{P}_{(n_\mathrm{r},n_\mathrm{c})}(B)}\PP(\kvec),$
    where $ \mathcal{P}_{(n_\mathrm{r},n_\mathrm{c})}(B)$ is the set of successful paths for a busy period $B$. Since the probability $\PP(\kvec)$ of any path is positive and does not depend on cutoff rule in the busy period, the result follows immediately from the inclusions of the successful paths in Lemma \ref{lemma: succesful path inclusion lemma}.
\end{proof}

We now aim to derive a similar result for the expected durations of the busy periods of $B'(c/2)$, $B(c)$, and $B'(c)$. 
The expected duration of the busy period $B'(c)$ is found as
\begin{equation}
    \EE(B'(c)) = \sum_{(i^*,j^*)\leq (n_\mathrm{r},n_\mathrm{c})} \sum_{\kvec \in \mathcal{P}_{(i^*,j^*)}(B'(c))}\PP(\kvec) B'(\kvec;c),
\end{equation}
where 
\begin{equation}\label{eq: duration B'}
    B'(\kvec;c) =
    \begin{cases}
        \sum_{j<j^*}\sum_{i=1}^{n_\mathrm{r}} K_{i,j} + c  & \mathrm{if\,}(i^*,j^*)<(n_\mathrm{r},n_\mathrm{c}) \\
        \sum_{(i,j)\leq (n_\mathrm{r},n_\mathrm{c})}K_{i,j}& \mathrm{if\,}(i^*,j^*)=(n_\mathrm{r},n_\mathrm{c}) 
    \end{cases}. 
\end{equation}
Again, $(i^*,j^*)$ can be determined from $\kvec$ since there is a unique qubit where the path terminates.

The idea is now that the duration of a path $\kvec\in \NN^{n-1}$ subject to the cutoff condition in $B(c)$ is always longer than the duration of that same path subject to the cutoff condition in $B'(c/2)$ since the cutoff of the latter is more stringent. Similar, paths in $B'(c)$ have longer durations than paths in $B(c)$ because again the latter has a more stringent cutoff. We make this precise below. 
\begin{proposition}[Expected duration busy period bounds]\label{prop: busy period bounds}
    The expected durations of the busy periods of $B'(c/2)$, $B(c)$, and $B'(c)$ satisfy
    \begin{equation}
        \EE(B'(c/2))\leq \EE(B(c)) \leq \EE(B'(c)).
    \end{equation}
\end{proposition}
\begin{proof}
    We start with the first inequality. If a path $\kvec\in \mathcal{P}_{(n_\mathrm{r},n_\mathrm{c})}(B'(c/2))$ so that it succeeds in $B'(c/2)$, then by Lemma \ref{lemma: succesful path inclusion lemma} it also succeeds in $B(c)$. The duration of $\kvec$ in $B'(c/2)$ is the same as in $B(c)$ since in both cases the duration is the sum of all times in the path (c.f. \cref{eq: duration B} and \cref{eq: duration B'}). Now let $\kvec \in \mathcal{P}_{(i^*,j^*)}(B'(c/2))$ with $(i^*,j^*)<(n_\mathrm{r},n_\mathrm{c})$ be a failed path in $B'(c/2)$ so that
    \begin{equation}
    \sum_{i'=1}^{n_\mathrm{r}}k_{i',j}\leq c/2 \quad\forall j<j^*,\quad \sum_{i'=1}^{i^*-1}k_{i',j^*} \leq c/2,\quad \mathrm{and}\quad \sum_{i'=1}^{i^*}k_{i',j^*}>c/2.
    \end{equation}
    It follows that for all $(i^\dagger,j^\dagger)<(i^*,j^*)$,
    \begin{equation}
        \sum_{i'=i^\dagger+1}^{n_\mathrm{r}}k_{i',j^\dagger-1} + \sum_{i'=1}^{i^\dagger}k_{i',j^\dagger}\leq c/2 + c/2 = c.
    \end{equation}
    By \cref{eq: condition fail B} and \cref{eq: condition succ B} it follows that $\kvec\in \mathcal{P}_{(i^\dagger, j^\dagger)}(B(c))$ for some $(i^\dagger,j^\dagger)\geq (i^*,j^*)$. If $(i^\dagger,j^\dagger)<(n_\mathrm{r},n_\mathrm{c})$, so that the path also fails in $B(c)$, then
    \begin{align}
        B'(\kvec;c/2) &= \sum_{j<j^*}\sum_{i=1}^{n_\mathrm{r}}k_{i,j} + c/2 \\
        &= \sum_{(i,j)
        \leq (i^*,j^*-1)} k_{i,j} + \sum_{i=i^*+1}^{n_\mathrm{r}}k_{i,j^*-1} + c/2 \\
        &\leq \sum_{(i,j)\leq(i^\dagger,j^\dagger-1)}k_{i,j} + c/2 + c/2 \\
        &\leq  B(\kvec;c),
    \end{align}
    where the first equation follows from \cref{eq: duration B'}, in the next step we split the sum in column $j^*-1$ anticipating the expression for the duration in $B(c)$, then for the first inequality we used \cref{eq: condition fail B'} to see that the times in column $j^*$ sum to less than $c/2$ and in the second inequality we used that $(i^*,j^*)\leq (i^\dagger, j^\dagger)<(n_\mathrm{r},n_\mathrm{c})$ and \cref{eq: duration B}. If, on the other hand, the path ends up succeeding in $B(c)$ so that $(i^\dagger,j^\dagger)=(n_\mathrm{r},n_\mathrm{c})$, then it follows from \cref{eq: condition fail B'} that
    \begin{align}
        B'(\kvec;c/2) &= \sum_{j<j^*}\sum_{i=1}^{n_\mathrm{r}}k_{i,j} + c/2 < \sum_{j<j^*}\sum_{i=1}^{n_\mathrm{r}}k_{i,j} + \sum_{i=1}^{i^*}k_{i,j^*} = \sum_{(i,j)\leq (i^*,j^*)}k_{i,j^*}\leq B(\kvec;c).
    \end{align}
    In both cases, the duration of the path $\kvec$ in $B(c)$ is thus seen to be longer than its duration in $B'(c/2)$. Hence, since the paths are weighed by the same probabilities irrespective of the cutoff condition, we conclude that $\EE(B'(c/2))\leq \EE(B(c))$.

    We now apply similar reasoning to establish the second inequality. First, if a path $\kvec$ is successful in $B(c)$, then by Lemma \ref{lemma: succesful path inclusion lemma} it is successful in $B'(c)$, and in both cases it has the same duration. If a path $\kvec\in \mathcal{P}_{(i^*,j^*)}(B(c))$ terminates at $(i^*,j^*)<(n_\mathrm{r},n_\mathrm{c})$, then it terminates at some later $(i^\dagger,j^\dagger)$ in $B'(c)$. Indeed, if it would terminate earlier this would mean by \cref{eq: condition fail B'} that $\sum_{i=1}^{i^\dagger}k_{i,j^\dagger}>c$ for some $(i^\dagger,j^\dagger)<(i^*,j^*)$. But that contradicts the fact that the cutoff condition \cref{eq: condition fail B} in $B(c)$ was not met for $(i^\dagger,j^\dagger)<(i^*,j^*)$. We conclude that $(i^\dagger,j^\dagger)\geq (i^*,j^*)$. If the path $\kvec$ ends up failing in $B'(c)$, i.e. if it is in $\mathcal{P}_{(i^\dagger,j^\dagger)}(B'(c))$ for some $(i^*, j^*)\leq (i^\dagger, j^\dagger)<(n_\mathrm{r}, n_\mathrm{c})$, then
    \begin{align}
        B(\kvec;c) & = \sum_{(i,j)\leq (i^*,j^*-1)} k_{i,j} + c \\
        &\leq \sum_{j<j^*}\sum_{i=1}^{n_\mathrm{r}}k_{i,j} + c \\
        &\leq \sum_{j<j^\dagger}\sum_{i=1}^{n_\mathrm{r}}k_{i,j} + c \\
        &=B'(\kvec;c),
    \end{align}
    in the first inequality we add the remaining attempts in column $j^*-1$, then in the second inequality we use that $j^*\leq j^\dagger$.
    If on the other hand the path ends up succeeding in $B'(c)$, so that $(i^\dagger,j^\dagger)=(n_\mathrm{r}, n_\mathrm{c})$, then
    \begin{align*}
        B(\kvec;c) &=\sum_{(i,j)\leq (i^*,j^*-1)} k_{i,j} + c \\
        &< \sum_{(i,j)\leq (i^*,j^*-1)} k_{i,j} + \sum_{i=i^*+1}^{n_\mathrm{r}} k_{i,j^*-1} + \sum_{i=1}^{i^*} k_{i,j^*}\\
        &\leq \sum_{(i,j)\leq (n_\mathrm{r},n_\mathrm{c})} k_{i,j}\\
        &= B'(\kvec;c),
    \end{align*}
    where in the first inequality we used \cref{eq: condition fail B} of the failed cutoff condition in $B(c)$. We thus conclude that the duration of any path $\kvec$ is longer in $B'(c)$ than in $B(c)$, so that $\EE(B(c))\leq \EE(B'(c))$. This completes the proof of both inequalities.
\end{proof}

\begin{corollary}[Bounds on expected time to create full graph]\label{cor: Bounds on expected time to create full graph}
    The expected time create the full graph with qubitwise cutoffs is bounded by quantities for columnwise cutoffs as
    \begin{equation}\label{eq: bounds on expected time equation statement}
        \frac{1/p+\EE(B'(c/2))}{p_{B'(c)}} \leq \EE(K(c)) \leq \frac{1/p+\EE(B'(c))}{p_{B'(c/2)}}.
    \end{equation}
\end{corollary}
\begin{proof}
    By \cref{eq: expectation K abstract} the expected number of attempts with qubitwise cutoffs is
    \begin{equation}
        \EE(K(c)) = \frac{1/p+\EE(B(c))}{p_{B(c)}}.
    \end{equation}
    The bounds on $\EE(K(c))$ follow directly from the bounds on $p_{B(c)}$ of \cref{cor: success probability bounds} and the bounds on $\EE(B(c))$ from \cref{prop: busy period bounds}.
\end{proof}

\subsection{Evaluation of the bounds for brickwork graphs}
The columnwise cutoffs were introduced to obtain bounds that these would be easier to evaluate than the exact result for qubitwise cutoffs because the columns are independent.
This gives the brickwork with columnwise cutoffs the same structure as a linear graph with qubitwise cutoffs. To make this precise, introduce the random variables $K_j'\in \{1,\dots, c\}$ for the number of attempts in column $j$ and $O_j'\in \{s,f\}$ for the outcome of column $j$. For $j<n_\mathrm{c}$ these follow the joint distribution, as explained directly after,
\begin{align}
    \PP(K_j'=k, O_j'=s) &= \binom{k-1}{n_\mathrm{r}-1}p^{n_\mathrm{r}}(1-p)^{k-n_\mathrm{r}}\qquad \mathrm{for}\quad k\in \{n_\mathrm{r},\dots,c\}\\
    \PP(K_j'=c, O_j'=f) &= \sum_{i^*=1}^{n_\mathrm{r}}\binom{c-1}{i^*-1}p^{i^*-1} (1-p)^{c-(i^*-1)}.
\end{align}
In the case $O_j'=s$ there are $n_\mathrm{r}$ successful generations in column $j$ in a total of $k\leq c$ attempts. There are a total of $n_\mathrm{r}$ successes and $k-n_\mathrm{r}$ failures. The $k$-th attempt in the column is successful. The first $k-1$ attempts must contain the other $n_\mathrm{r}-1$ successes. On the other hand, in the case of $O_j'=f$ the column fails so there must have been a total of $c$ attempts with strictly less than $n_\mathrm{r}$ successes. If the column fails at row $i^*$, then there were $i^*-1$ successes. The $c$-th attempt must be a failure, so these $i^*-1$ successes must come from the first $c-1$ attempts.
The distribution of $(K'_{n_\mathrm{c}},O'_{n_\mathrm{c}})$ for the last column $j=n_\mathrm{c}$ is similar. However, in this column only $n_\mathrm{r}-1$ successes are required since no further generation is needed at qubit $(n_\mathrm{r}, n_\mathrm{c})$. Hence, the joint distribution $(K'_{n_\mathrm{c}},O'_{n_\mathrm{c}})$ is obtained by using $n_\mathrm{r}-1$ instead of $n_\mathrm{r}$ in the above, so that
\begin{align}
    \PP(K'_{n_\mathrm{c}}=k, O'_{n_\mathrm{c}}=s) &= \binom{k-1}{n_\mathrm{r}-2}p^{n_\mathrm{r}-1}(1-p)^{k-(n_\mathrm{r}-1)}\qquad \mathrm{for}\quad k\in \{n_\mathrm{r}-1,\dots,c\}\\
    \PP(K'_{n_\mathrm{c}}=c, O'_{n_\mathrm{c}}=f) &= \sum_{i^*=1}^{n_\mathrm{r}-1}\binom{c-1}{i^*-1}p^{i^*-1} (1-p)^{c-(i^*-1)}.
\end{align}
With this, the probability generating function for the number of attempts to create the full brickwork graph with columnwise cutoffs is of the form \cref{eq: generating function linear graph}.  
\begin{proposition}\label{prop: generating fucntion brickwork graph}
    The probability generating function for the duration of the busy period $B'$ for the brickwork graph with $n_\mathrm{r}>1$ rows and columnwise cutoffs is given by
    \begin{equation}\label{eq: generating function brickwork}
        G_{B'}(x) = S_1'(x)\dots S'_{n_\mathrm{c}}(x)+  \sum_{j^*=1}^{n_\mathrm{c}}S'_1(x)\dots S'_{j^*-1}(x)F'_{j^*}(x),
    \end{equation}
    where 
    \begin{equation}\label{eq: definition defective generating function brickwork graph}
    S'_j(x):= \EE\big(x^{K'_j}\indicator(O'_j=s)\big) \qquad \mathrm{and}\qquad F'_j(x):= \EE\big(x^{K'_j}\indicator(O'_j=f)\big).
    \end{equation}
    are the defective generating functions for $K'_j$ for $O'_j=s$ or $O'_j=f$, respectively.
\end{proposition}
\begin{proof} 
The proof is identical in structure to that of \cref{prop: generating fucntion linear graph}. The only difference is that here the graph generation can fail in column $n_\mathrm{c}$ so that $B'$ contains a term $K'_{n_\mathrm{c}}$ and the sum over $j^*$ in \cref{eq: generating function brickwork} runs through to $j^*=n_\mathrm{c}$.
\end{proof}
Closed-form expressions for the defective generating functions $S'_j(x)$ and $F'_j(x)$ of the brickwork graph with columnwise cutoffs can be found explicitly from the joint distributions for $(K'_j,O'_j)$. We remark that these generating functions are in fact similar to those found in the analysis of sequential repeater chains with cutoffs \cite{kamin2023exact}. 

The success probability of the busy period with columnwise cutoffs is
\begin{equation}
    p_{B'} = S'_1(1)\cdots S'_{n_\mathrm{c}}(1).
\end{equation}
To lowest order in $p$ we have 
\begin{equation}
    S'_j(1) = \sum_{k=n_\mathrm{r}}^c \binom{k-1}{n_\mathrm{r}-1}p^{n_\mathrm{r}}(1-p)^{k-n_\mathrm{r}} = \sum_{k=n_\mathrm{r}}^c \binom{k-1}{n_\mathrm{r}-1}p^{n_\mathrm{r}} + O(p^{n_\mathrm{r}+1}) = \binom{c}{n_\mathrm{r}}p^{n_\mathrm{r}} + O(p^{n_\mathrm{r}+1}),
\end{equation}
where the last equality uses the hockey-stick identity for the binomial coefficient. 
Similarly, for $j=n_\mathrm{c}$ we find
\begin{equation}
    S'_{n_\mathrm{c}}(1) = \binom{c}{n_\mathrm{r}-1}p^{n_\mathrm{r}-1} + O(p^{n_\mathrm{r}}).
\end{equation}
Hence, the lowest order $p$-scaling of the success probability for the busy period with columnwise cutoff $c$ is
\begin{equation}
    p_{B'(c)} = \binom{c}{n_\mathrm{r}}^{n_\mathrm{c}-1}\binom{c}{n_\mathrm{r}-1} p^{n-1}+ O(p^n).
\end{equation}
The duration of the busy period is upperbounded independently of $p$ by $B'(c)\leq n_\mathrm{c}c$ since every column can take at most $c$ attempts. The scaling of expected number of attempts to create the brickwork graph with columnwise cutoff $c$ is thus
\begin{equation}\label{eq: lowest order p scaling columnwise cutoff brickwork}
    \EE\Big(K'(c)\Big)\lesssim \frac{1/p + n_\mathrm{c}c}{\binom{c}{n_\mathrm{r}}^{n_\mathrm{c}-1}\binom{c}{n_\mathrm{r}-1}p^{n-1}}.
\end{equation}
For $1/p\gg n_\mathrm{c}c$ and $c\gg n_\mathrm{r}$, this yields the scalings in \cref{eq: brickwork scaling column}.

The bounds in \cref{cor: Bounds on expected time to create full graph} can also be evaluated numerically. For the \textit{Mathematica} notebook that does this we refer to our repository at \url{https://gitlab.tudelft.nl/wehner-research/bqc_graph_state_creation}. In \cref{fig:examples bounds} we show examples of the bounds the expected number of attempts to generate brickwork graph states. In the parameter regime shown, the upper and lowerbounds obtained from the columnwise cutoffs are tight for large values of $p$, but differ by multiple orders of magnitude for smaller values of $p$.

\begin{figure}
  \centering
  \begin{subfigure}{0.32\textwidth}
    \centering
    \includegraphics[width=\linewidth]{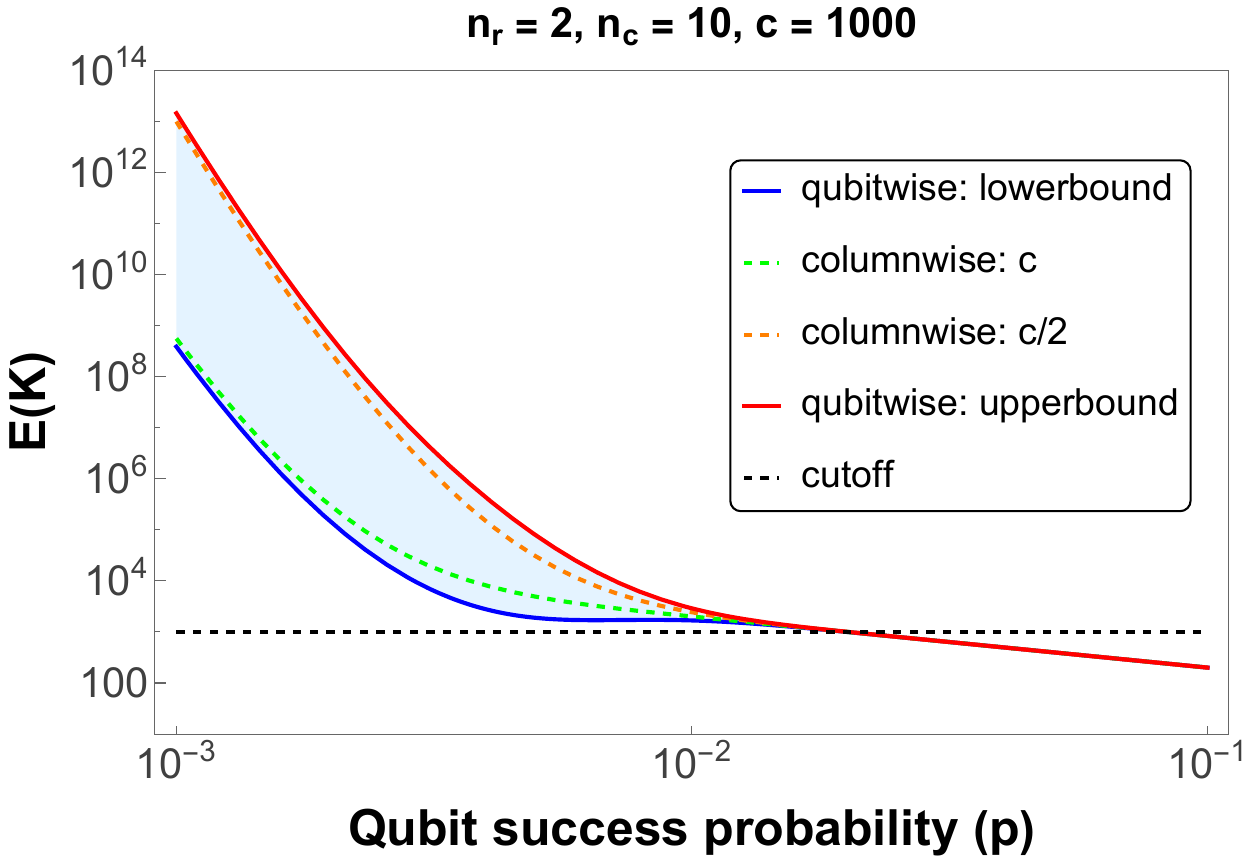}
    \caption{}
    \label{fig:nr2nc10c1000_incl_column}
  \end{subfigure}
  \hfill
  \begin{subfigure}{0.32\textwidth}
    \centering
    \includegraphics[width=\linewidth]{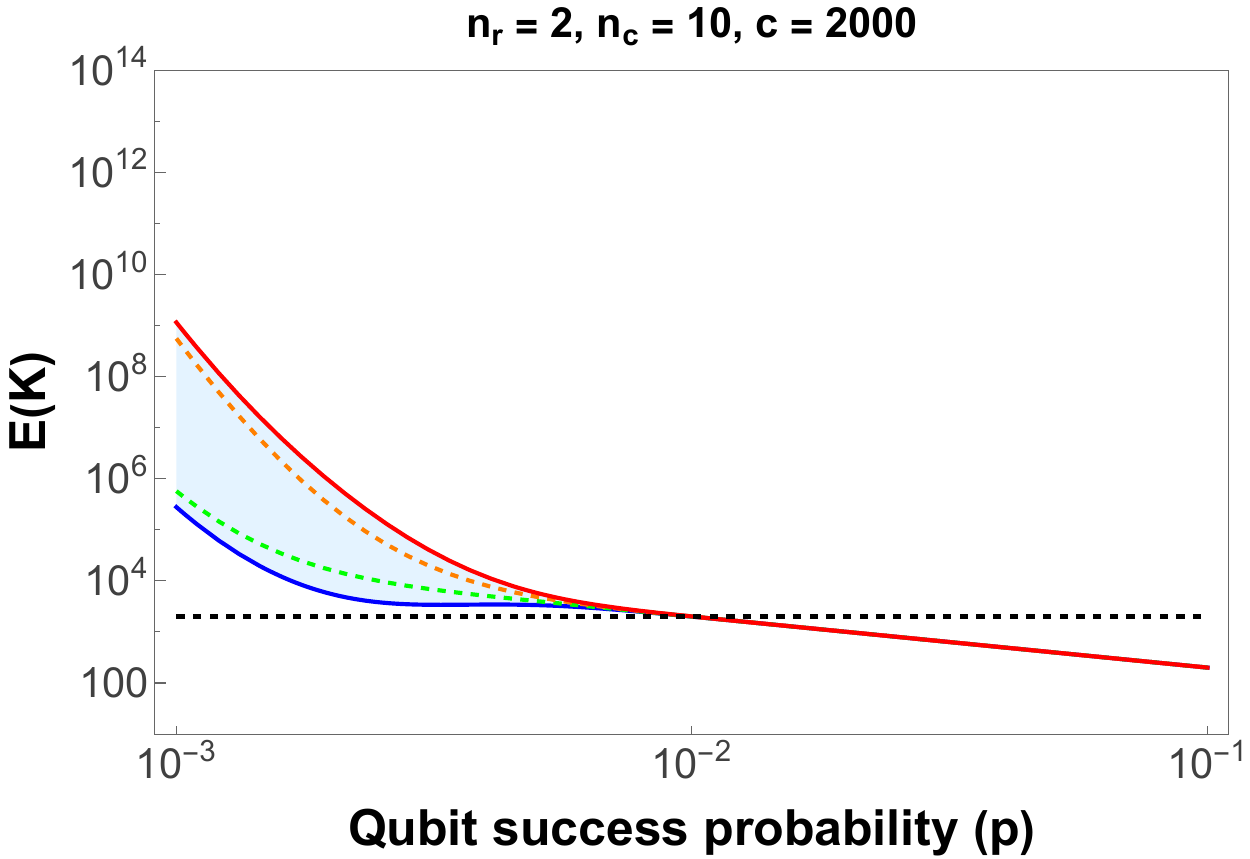}
    \caption{}
    \label{fig:nr2nc10c2000}
  \end{subfigure}
  \hfill
  \begin{subfigure}{0.32\textwidth}
    \centering
    \includegraphics[width=\linewidth]{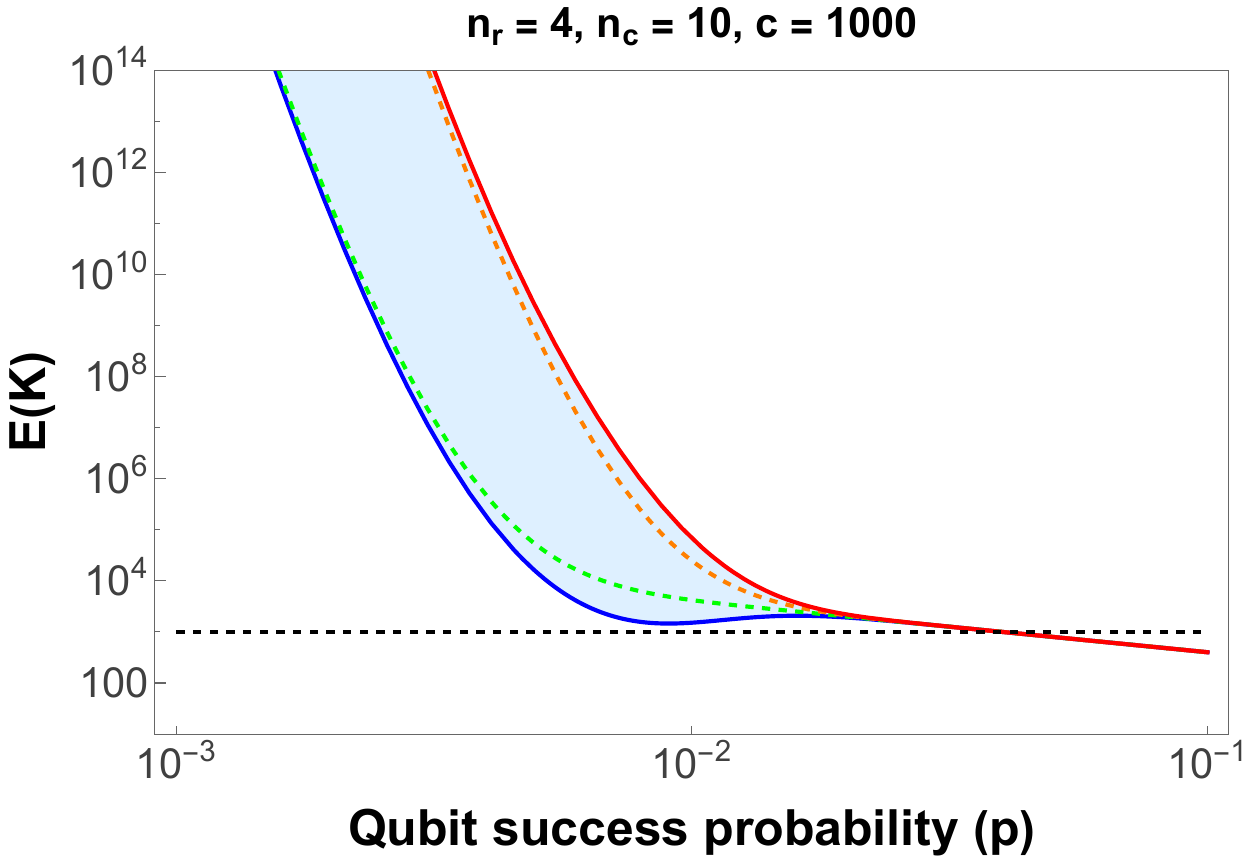}
    \caption{}
    \label{fig:nr4nc10c1000}
  \end{subfigure}
  \caption{Expected number of attempts for generating brickwork graph state with cutoffs. Legend of (a) applies to (b) and (c) as well. (a) We show the upper and lowerbounds to the expected number of attempts to generate the brickwork graph state with qubitwise cutoffs $n_\mathrm{r}=2$ and $n_\mathrm{c}=10$ with cutoff $c=1000$ in the regime $p\in [10^{-3},10^{-1}]$, as in \cref{fig:nr2nc10c1000}. The cutoff $c$ is shown as a horizontal dashed line. We also include the expected number of attempts for columnwise cutoffs with cutoff $c$ and $c/2$ explicitly. It can be seen that as $p$ decreases the upperbound becomes similar to a columnwise cutoff of $c/2$ and the lowerbound becomes similar to a columnwise cutoff of $c$. In (b) the cutoff is increased to $c=2000$. In (c) the number of rows in the graph state is increased to $n_\mathrm{r}=4$. In all cases, the expected time to generate the graph state blows up when $p$ is no longer large compared to $n_\mathrm{r}/c$. Comparing (a) to (b), it can be seen that when the cutoff is increased, the number of attempts decreases. On the other hand, comparing (a) to (c), it can be seen that when the number of rows is increased the expected number of attempts increases. We also notice that the lowerbound is not monotonic in $p$, which indicates that it can be improved further.}
  \label{fig:examples bounds}
\end{figure}

%% file: sourcefile_manuscript_adjusted.bbl
\begin{thebibliography}{10}

\bibitem{broadbent2009universal}
Anne Broadbent, Joseph Fitzsimons, and Elham Kashefi.
\newblock Universal blind quantum computation.
\newblock In {\em 2009 50th annual IEEE symposium on foundations of computer
  science}, pages 517--526. IEEE, 2009.

\bibitem{fitzsimons2017unconditionally}
Joseph~F Fitzsimons and Elham Kashefi.
\newblock Unconditionally verifiable blind quantum computation.
\newblock {\em Physical Review A}, 96(1):012303, 2017.

\bibitem{morimae2013secure}
Tomoyuki Morimae and Keisuke Fujii.
\newblock Secure entanglement distillation for double-server blind quantum
  computation.
\newblock {\em Physical Review Letters}, 111:020502, Jul 2013.

\bibitem{sheng2015deterministic}
Yu-Bo Sheng and Lan Zhou.
\newblock Deterministic entanglement distillation for secure double-server
  blind quantum computation.
\newblock {\em Scientific reports}, 5(1):7815, 2015.

\bibitem{li2014triple}
Qin Li, Wai~Hong Chan, Chunhui Wu, and Zhonghua Wen.
\newblock Triple-server blind quantum computation using entanglement swapping.
\newblock {\em Physical Review A}, 89(4):040302, 2014.

\bibitem{quan2023verifiable}
Junyu Quan, Qin Li, and Lvzhou Li.
\newblock Verifiable blind quantum computation with identity authentication for
  multi-type clients.
\newblock {\em IEEE Transactions on Information Forensics and Security}, 2023.

\bibitem{qline}
Beatrice Polacchi, Dominik Leichtle, Leonardo Limongi, Gonzalo Carvacho,
  Giorgio Milani, Nicol{\`o} Spagnolo, Marc Kaplan, Fabio Sciarrino, and Elham
  Kashefi.
\newblock Multi-client distributed blind quantum computation with the qline
  architecture.
\newblock {\em Nature Communications}, 14(1):7743, 2023.

\bibitem{bruzewicz2019trapped}
Colin~D. Bruzewicz, John Chiaverini, Robert McConnell, and Jeremy~M. Sage.
\newblock Trapped-ion quantum computing: Progress and challenges.
\newblock {\em Applied Physics Reviews}, 6(2):021314, 2019.

\bibitem{doherty2013nitrogen}
Marcus~W. Doherty, Neil~B. Manson, Paul Delaney, Fedor Jelezko, J{\"o}rg
  Wrachtrup, and Lloyd C.~L. Hollenberg.
\newblock The nitrogen-vacancy colour centre in diamond.
\newblock {\em Physics Reports}, 528(1):1--45, 2013.

\bibitem{henriet2020quantum}
Lo{\"\i}c Henriet, Lucas Beguin, Adrien Signoles, Thierry Lahaye, Antoine
  Browaeys, Georges-Olivier Reymond, and Christophe Jurczak.
\newblock Quantum computing with neutral atoms.
\newblock {\em Quantum}, 4:327, 2020.

\bibitem{mahadev2018classical}
Urmila Mahadev.
\newblock Classical verification of quantum computations.
\newblock In {\em 2018 IEEE 59th Annual Symposium on Foundations of Computer
  Science (FOCS)}, pages 259--267, 2018.

\bibitem{broadbent2015quantum}
Anne Broadbent and Stacey Jeffery.
\newblock Quantum homomorphic encryption for circuits of low {T}-gate
  complexity.
\newblock In {\em Advances in Cryptology--CRYPTO 2015}, pages 609--629.
  Springer, 2015.

\bibitem{ma2022qenclave}
Yao Ma, Elham Kashefi, Myrto Arapinis, Kaushik Chakraborty, and Marc Kaplan.
\newblock Qenclave-a practical solution for secure quantum cloud computing.
\newblock {\em npj Quantum Information}, 8(1):128, 2022.

\bibitem{raussendorf2001one}
Robert Raussendorf and Hans~J Briegel.
\newblock A one-way quantum computer.
\newblock {\em Physical Review Letters}, 86(22):5188, 2001.

\bibitem{hein2004multiparty}
Marc Hein, Jens Eisert, and Hans~J Briegel.
\newblock Multiparty entanglement in graph states.
\newblock {\em Physical Review A—Atomic, Molecular, and Optical Physics},
  69(6):062311, 2004.

\bibitem{childs2005secure}
Andrew~M Childs.
\newblock Secure assisted quantum computation.
\newblock {\em Quantum Information \& Computation}, 5(6):456--466, 2005.

\bibitem{gheorghiu2019verification}
Alexandru Gheorghiu, Theodoros Kapourniotis, and Elham Kashefi.
\newblock Verification of quantum computation: An overview of existing
  approaches.
\newblock {\em Theory of computing systems}, 63(4):715--808, 2019.

\bibitem{morimae2014verification}
Tomoyuki Morimae.
\newblock Verification for measurement-only blind quantum computing.
\newblock {\em Physical Review A}, 89(6):060302(R), 2014.

\bibitem{hayashi2015verifiable}
Masahito Hayashi and Tomoyuki Morimae.
\newblock Verifiable measurement-only blind quantum computing with stabilizer
  testing.
\newblock {\em Physical review letters}, 115(22):220502, 2015.

\bibitem{eberhard1989quantum}
Phillippe~H Eberhard and Ronald~R Ross.
\newblock Quantum field theory cannot provide faster-than-light communication.
\newblock {\em Foundations of Physics Letters}, 2(2):127--149, 1989.

\bibitem{vandam2025single}
Janice van Dam, Emil~R Hellebek, Tzula~B Propp, Junior~R Gonzales-Ureta,
  Anders~S S{\o}rensen, and Stephanie~DC Wehner.
\newblock Single-click protocols for remote state preparation using weak
  coherent pulses.
\newblock {\em arXiv preprint arXiv:2508.14857}, 2025.

\bibitem{barrett2005efficient}
Sean~D Barrett and Pieter Kok.
\newblock Efficient high-fidelity quantum computation using matter qubits and
  linear optics.
\newblock {\em Physical Review A—Atomic, Molecular, and Optical Physics},
  71(6):060310, 2005.

\bibitem{campbell2008measurement}
Earl~T Campbell and Simon~C Benjamin.
\newblock Measurement-based entanglement under conditions of extreme photon
  loss.
\newblock {\em Physical review letters}, 101(13):130502, 2008.

\bibitem{pinotsi2008single}
Dorothea Pinotsi and Atac Imamoglu.
\newblock Single photon absorption by a single quantum emitter.
\newblock {\em Physical review letters}, 100(9):093603, 2008.

\bibitem{rosenblum2017analysis}
Serge Rosenblum, Adrien Borne, and Barak Dayan.
\newblock Analysis of deterministic swapping of photonic and atomic states
  through single-photon raman interaction.
\newblock {\em Physical Review A}, 95(3):033814, 2017.

\bibitem{duan2004scalable}
L-M Duan and HJ~Kimble.
\newblock Scalable photonic quantum computation through cavity-assisted
  interactions.
\newblock {\em Physical review letters}, 92(12):127902, 2004.

\bibitem{kikura2025passive}
Seigo Kikura, Kazufumi Tanji, Akihisa Goban, and Shinichi Sunami.
\newblock Passive quantum interconnects: High-fidelity quantum networking at
  higher rates with less overhead.
\newblock {\em arXiv preprint arXiv:2507.01229}, 2025.

\bibitem{vandam2024hardware}
Janice van Dam, Guus Avis, Tz~B Propp, Francisco Ferreira~da Silva, Joshua~A
  Slater, Tracy~E Northup, and Stephanie Wehner.
\newblock Hardware requirements for trapped-ion-based verifiable blind quantum
  computing with a measurement-only client.
\newblock {\em Quantum Science and Technology}, 9(4):045031, 2024.

\bibitem{kashefi2024verification}
Elham Kashefi, Dominik Leichtle, Luka Music, and Harold Ollivier.
\newblock Verification of quantum computations without trusted preparations or
  measurements.
\newblock {\em arXiv preprint arXiv:2403.10464}, 2024.

\bibitem{li2021blind}
Qin Li, Chengdong Liu, Yu~Peng, Fang Yu, and Cai Zhang.
\newblock Blind quantum computation where a user only performs single-qubit
  gates.
\newblock {\em Optics \& Laser Technology}, 142:107190, 2021.

\bibitem{wu2023blind}
Guang-Yang Wu, Zhen Yang, Yu-Zhan Yan, Yuan-Mao Luo, Ming-Qiang Bai, and
  Zhi-Wen Mo.
\newblock Blind quantum computation with a client performing different
  single-qubit gates.
\newblock {\em Chinese Physics B}, 32(11):110302, 2023.

\bibitem{takeuchi2018verification}
Yuki Takeuchi and Tomoyuki Morimae.
\newblock Verification of many-qubit states.
\newblock {\em Phys. Rev. X}, 8:021060, 2018.

\bibitem{leichtle2021verifying}
Dominik Leichtle, Luka Music, Elham Kashefi, and Harold Ollivier.
\newblock Verifying bqp computations on noisy devices with minimal overhead.
\newblock {\em PRX Quantum}, 2(4):040302, 2021.

\bibitem{kapourniotis2024unifying}
Theodoros Kapourniotis, Elham Kashefi, Dominik Leichtle, Luka Music, and Harold
  Ollivier.
\newblock Unifying quantum verification and error-detection: theory and tools
  for optimisations.
\newblock {\em Quantum Science and Technology}, 9(3):035036, 2024.

\bibitem{dunjko2014composable}
Vedran Dunjko, Joseph~F Fitzsimons, Christopher Portmann, and Renato Renner.
\newblock Composable security of delegated quantum computation.
\newblock In {\em International Conference on the Theory and Application of
  Cryptology and Information Security}, pages 406--425. Springer, 2014.

\bibitem{garnier2024composably}
Maxime Garnier, Dominik Leichtle, Luka Music, and Harold Ollivier.
\newblock Composably secure delegated quantum computation with weak coherent
  pulses.
\newblock In {\em 2024 International Conference on Quantum Communications,
  Networking, and Computing (QCNC)}, pages 221--225. IEEE, 2024.

\bibitem{kapourniotis2023asymmetric}
Theodoros Kapourniotis, Elham Kashefi, Dominik Leichtle, Luka Music, and Harold
  Ollivier.
\newblock Asymmetric quantum secure multi-party computation with weak clients
  against dishonest majority.
\newblock {\em arXiv preprint arXiv:2303.08865}, 2023.

\bibitem{vandam2026optimizing}
Janice van Dam, Michal van Hooft, and Stephanie~D.C. Wehner.
\newblock Optimizing resource costs: A practical guide to achieving target
  security in verifiable blind quantum computing.
\newblock {\em arXiv preprint arXiv:2606.28139}, 2026.

\bibitem{krutyanskiy2023entanglement}
Viktor Krutyanskiy, Maria Galli, Vojtech Krcmarsky, Simon Baier, DA~Fioretto,
  Yunfei Pu, Azadeh Mazloom, Pavel Sekatski, Marco Canteri, Markus Teller,
  et~al.
\newblock Entanglement of trapped-ion qubits separated by 230 meters.
\newblock {\em Physical Review Letters}, 130(5):050803, 2023.

\bibitem{danos2007measurement}
Vincent Danos, Elham Kashefi, and Prakash Panangaden.
\newblock The measurement calculus.
\newblock {\em Journal of the ACM (JACM)}, 54(2):8--es, 2007.

\bibitem{ross1995stochastic}
Sheldon~M Ross.
\newblock {\em Stochastic processes}.
\newblock John Wiley \& Sons, 1995.

\bibitem{reiher2017elucidating}
Markus Reiher, Nathan Wiebe, Krysta~M Svore, Dave Wecker, and Matthias Troyer.
\newblock Elucidating reaction mechanisms on quantum computers.
\newblock {\em Proceedings of the national academy of sciences},
  114(29):7555--7560, 2017.

\bibitem{propp2025quantum}
Tzula~B Propp, Jeroen Grimbergen, Emil~R Hellebek, Junior~R Gonzales-Ureta,
  Janice van Dam, Joshua~A Slater, Anders~S S{\o}rensen, and Stephanie~DC
  Wehner.
\newblock Quantum strategies to overcome classical multiplexing limits.
\newblock {\em arXiv preprint arXiv:2510.06099}, 2025.

\bibitem{legero2003time}
Thomas Legero, Tatjana Wilk, Axel Kuhn, and Gerhard Rempe.
\newblock Time-resolved two-photon quantum interference.
\newblock {\em Applied Physics B}, 77(8):797--802, 2003.

\bibitem{meyer2015direct}
HM~Meyer, R~Stockill, M~Steiner, C~Le~Gall, Clemens Matthiesen, E~Clarke,
  A~Ludwig, J~Reichel, M~Atat{\"u}re, and M~K{\"o}hl.
\newblock Direct photonic coupling of a semiconductor quantum dot and a trapped
  ion.
\newblock {\em Physical review letters}, 114(12):123001, 2015.

\bibitem{purcell1946resonance}
Edward~M Purcell, Henry~Cutler Torrey, and Robert~V Pound.
\newblock Resonance absorption by nuclear magnetic moments in a solid.
\newblock {\em Physical review}, 69(1-2):37, 1946.

\bibitem{drmota2024verifiable}
Peter Drmota, David~P Nadlinger, Dougal Main, Bethan~C Nichol, Ellis~M Ainley,
  Dominik Leichtle, Atul Mantri, Elham Kashefi, Raghavendra Srinivas, Gabriel
  Araneda, Chris~J Ballance, and David~M Lucas.
\newblock Verifiable blind quantum computing with trapped ions and single
  photons.
\newblock {\em Physical Review Letters}, 132(15):150604, 2024.

\bibitem{mahadev2020classical}
Urmila Mahadev.
\newblock Classical homomorphic encryption for quantum circuits.
\newblock {\em SIAM Journal on Computing}, 52(6):FOCS18--189, 2020.

\bibitem{gheorghiu2019computationally}
Alexandru Gheorghiu and Thomas Vidick.
\newblock Computationally-secure and composable remote state preparation.
\newblock In {\em 2019 IEEE 60th Annual Symposium on Foundations of Computer
  Science (FOCS)}, pages 1024--1033. IEEE, 2019.

\bibitem{davies2023tools}
Bethany Davies, Thomas Beauchamp, Gayane Vardoyan, and Stephanie Wehner.
\newblock Tools for the analysis of quantum protocols requiring state
  generation within a time window.
\newblock {\em arXiv preprint arXiv:2304.12673}, 2023.

\bibitem{kamin2023exact}
Lars Kamin, Evgeny Shchukin, Frank Schmidt, and Peter Van~Loock.
\newblock Exact rate analysis for quantum repeaters with imperfect memories and
  entanglement swapping as soon as possible.
\newblock {\em Physical Review Research}, 5(2):023086, 2023.

\end{thebibliography}
